\documentclass[twocolumn,showpacs,amsmath,amssymb,aps,prx,footinbib,floatfix,superscriptaddress,longbibliography,nofootinbib]{revtex4-1}
\usepackage{graphicx}
\usepackage{braket}
\usepackage{savesym}
\usepackage{amsfonts}
\usepackage{bm}
\usepackage{color}
\usepackage{subfigure}
\usepackage{wasysym}
\usepackage{upgreek}
\usepackage{float}
\usepackage{dsfont}
\usepackage{mathtools}% Loads amsmath
\usepackage{multirow}
\usepackage[usenames,dvipsnames,svgnames,table]{xcolor}
\usepackage{hyperref}
\usepackage{hhline} %for double horizontal lines without vertical break in tabular environment
\usepackage{multirow} %for cells that span multiple rows in columns
\usepackage[normalem]{ulem}
\usepackage{amsthm}

\newtheorem{lemma}{Lemma}
\newtheorem{corollary}{Corollary}

\theoremstyle{remark}

\usepackage[english]{babel}
\usepackage{makecell}

\hypersetup{
	colorlinks=true,       % false: boxed links; true: colored links
	linkcolor=VioletRed,  % color of internal links (change box color with linkbordercolor)
    citecolor=JungleGreen,        % color of links to bibliography
    filecolor=Orchid,      % color of file links
    urlcolor=black           % color of external links
}
\begin{document}

% \title{Dilated Recurrent Neural Network Quantum States for Long-range Correlations}

% \title{Geometric Control of Long-Range Correlations in Dilated Recurrent Neural Network Quantum States}

\title{Geometry-Induced Long-Range Correlations in Recurrent Neural Network Quantum States}

% \title{Long-Range Correlations from Geometry in Recurrent Neural Network Quantum States}

%Correlation-Aware Neural Quantum States via Dilated Recurrent Architectures

%Architectural Control of Long-Range Correlations in Autoregressive Neural Quantum States

%Geometric Control of Long-Range Correlations in Recurrent Neural Quantum States

%The Geometry of Long-Range Correlations in Dilated Recurrent Neural Quantum States

\author{Asif Bin Ayub}
\affiliation{Department of Applied Mathematics, University of Waterloo, Waterloo, ON N2L 3G1, Canada}
\affiliation{Perimeter Institute for Theoretical Physics, 31 Caroline St N, Waterloo, ON N2L 2Y5, Canada}

\author{Amine M. Aboussalah}
\affiliation{NYU Tandon School of Engineering, 6 MetroTech Center, Brooklyn, New York 11201, United States}
\author{Mohamed Hibat-Allah}
\email{mohamed.hibatallah@uwaterloo.ca}
\affiliation{Department of Applied Mathematics, University of Waterloo, Waterloo, ON N2L 3G1, Canada}
\affiliation{Vector Institute,  Toronto,  Ontario,  M5G 0C6,  Canada}

\date{\today}

% %%%%%%%%%%%%%%%%%% ABSTRACT %%%%%%%%%%%%%%%%%%
\begin{abstract}
Neural Quantum States based on autoregressive recurrent neural network (RNN) wave functions enable efficient sampling without Markov-chain autocorrelation, but standard RNN architectures are biased toward finite-length correlations and can fail on states with long-range dependencies. A common response is to adopt transformer-style self-attention, but this typically comes with substantially higher computational and memory overhead. Here we introduce dilated RNN wave functions, where recurrent units access distant sites through dilated connections, injecting an explicit long-range inductive bias while retaining a favorable $\mathcal{O}(N \log N)$ forward pass scaling. We show analytically that dilation changes the correlation geometry and can induce power-law correlation scaling in a simplified linearized and perturbative setting. Numerically, for the critical 1D transverse-field Ising model, dilated RNNs reproduce the expected power-law connected two-point correlations in contrast to the exponential decay typical of conventional RNN ansätze. We further show that the dilated RNN accurately approximates the one-dimensional Cluster state, a paradigmatic example with long-range conditional correlations that has previously been reported to be challenging for RNN-based wave functions. These results highlight dilation as a simple geometric mechanism for building correlation-aware autoregressive neural quantum states.    
\end{abstract}

\maketitle

%%%%%%%%%%%%%%%%%% Introduction %%%%%%%%%%%%%%%%%%

\section{Introduction}

Neural Quantum States (NQS) have emerged as a powerful variational ansatz for the study of quantum many-body systems, offering a flexible and complementary alternative to traditional numerical techniques such as tensor networks and quantum Monte Carlo methods~\cite{ANDROSIUK1993377,LAGARIS19971,Carleo2017-NN_VMC, Lan24}. Among the various NQS architectures, Recurrent Neural Network (RNN) wave functions~\cite{Hibat-Allah2020-RNN, roth2020iterativeretrainingquantumspin, PhysRevResearch.5.013216, Lange_2024} have attracted particular attention due to their favorable computational scaling, structural versatility, and their relationship to language models~\cite{Melko2024, Malica_2025}. By factorizing the many-body wave function into a product of conditional probabilities, RNN wave functions enable efficient sampling and evaluation with a forward pass computational cost that scales linearly with system size~\cite{Hibat-Allah2020-RNN}. Moreover, their autoregressive nature allows for iterative retraining, making it possible to extrapolate models trained on smaller systems to larger system sizes~\cite{roth2020iterativeretrainingquantumspin, 6ccd-wzhz, nh89-6jmf}. More broadly, this reflects a recurring theme in sequential learning: preserving the relevant dynamical structure of the problem can substantially improve extrapolation~\cite{Aboussalah2023RIM}. These properties, together with the ease of generalization to higher spatial dimensions~\cite{Hibat-Allah2020-RNN, hibatallah2024supplementingrecurrentneuralnetwork} and exotic lattices~\cite{Hibat_Allah_2025}, make RNN wave functions an appealing tool for large-scale quantum simulation.

Despite these advantages, RNN-based NQS face notable challenges. In particular, standard RNN architectures are biased toward finite-length correlations. Recent work has clarified the role of correlation structure in variational neural-network-based simulations. In particular, Ref.~\cite{YangPreskill2024-ClusterStatePaper} demonstrated that the presence of long-range conditional correlations can severely inhibit the ability of RNN wave functions to approximate certain quantum ground states. Complementing this perspective, Ref.~\cite{doschl2025importancecorrelationsneuralquantum} investigated how correlations are represented internally in neural quantum states and showed that even simple quantum states can require access to high-order correlations in the network representation. Taken together, these findings highlight that the difficulty of variational neural-network-based simulations is closely tied to how correlations are encoded in the ansatz, rather than to representational expressivity alone. This broader viewpoint is consistent with recent results in machine learning showing that the way information is structured and propagated within a model can be as important for representation quality and generalization as the overall expressive power of the architecture~\cite{Aboussalah2023RIM,Aboussalah2025GNNTopology}.

In this work, we show that these limitations can be addressed by introducing dilated connections into RNNs~\cite{Chang2017DilatedRNN, Hibat-Allah2021-VNA, Khandoker_2023}. By allowing recurrent units to directly access information from distant sites, dilated RNNs provide an explicit architectural mechanism for encoding long-range dependencies, effectively injecting a long-range inductive bias while preserving a favorable computational scaling compared with attention-based alternatives.

We first present an analytical argument based on a linearized RNN model, showing that dilated connections modify the effective correlation geometry of the ansatz and can yield power-law correlation scaling. While standard RNN constructions are generically associated with exponentially decaying correlations~\cite{shen2019mutualinformationscalingexpressive, YangPreskill2024-ClusterStatePaper}, dilated RNNs admit regimes in which the first-order approximate connected two-point correlation function is bounded from below by a power-law. This establishes a direct link between architectural design and the emergent correlation structure of autoregressive neural quantum states.

Guided by this theoretical insight, we then investigate the performance of dilated RNN wave functions in numerical simulations of prototypical many-body systems. For the critical one-dimensional transverse-field Ising model (TFIM), we show that dilated RNNs accurately reproduce the expected power-law decay of connected two-point correlations and the associated critical scaling behavior, in contrast to the exponential decay typically observed in standard RNN constructions. Finally, we apply our approach to the one-dimensional Cluster state, a paradigmatic example of a system exhibiting long-range conditional correlations. While previous studies have reported significant difficulties in approximating this ground state using RNNs~\cite{YangPreskill2024-ClusterStatePaper, McNaughton2025-ku}, we show that our dilated RNN wave function succeeds at finding this ground state. Taken together, these results position dilation as a principled and computationally efficient approach for constructing correlation-aware autoregressive neural quantum states. More generally, our use of dilation reflects a broader strategy of incorporating structural priors into the model architecture to improve representation quality~\cite{Aboussalah2025GeoHNN,Aboussalah2025GNNTopology}.

%%%%%%%%%%%%%%%%%% Methods %%%%%%%%%%%%%%%%%%

\section{Methods}
\label{sec:methods}

% Introducing the general class of NQS
\subsection{Autoregressive Neural Quantum States}

NQS are a class of variational wave functions $\ket{\Psi_{\bm{\theta}}}$ that represent a system of $N$ qubits with a neural network ansatz~\cite{Carleo2017-NN_VMC, Lan24}, whose weights and biases $\{\bm{\theta}\}$ form the space of variational parameters. In this study, we restrict our attention to spin-$\frac{1}{2}$ systems and choose the computational basis $\boldsymbol{\sigma}=\{-1,1\}^N$. The NQS $\ket{\Psi_{\bm{\theta}}}$ is modeled, over $2^N$ possible $N$-qubit states $\ket{\boldsymbol{\sigma}}$, as:
\begin{align}
\begin{split}
    \ket{\Psi_{\bm{\theta}}} &= \sum_{\boldsymbol{\sigma}} \psi_{\bm{\theta}}({\boldsymbol{\sigma}})
    \ket{\boldsymbol{\sigma}} \\
    \psi_{\bm{\theta}}(\boldsymbol{\sigma}) &= \sqrt{P_{\bm{\theta}}(\boldsymbol{\sigma})}e^{i\phi_{\bm{\theta}}(\boldsymbol{\sigma})} ,
    \label{eq:amp}
\end{split}
\end{align}
where the amplitudes $\psi_{\bm{\theta}}(\boldsymbol{\sigma}) \in \mathbb{C}$, with phases $\phi(\boldsymbol{\sigma}) \in [-\pi,\pi]$ and probabilities $P_{\bm{\theta}}(\boldsymbol{\sigma})$, such that $\sum_{\boldsymbol{\sigma}}P_{\bm{\theta}}(\boldsymbol{\sigma})=1$.

Optimization of the parameters $\{{\bm{\theta}}\}$ is achieved through Variational Monte Carlo (VMC)~\cite{becca_sorella_2017}, which involves sampling from the RNN underlying distribution $P_{\bm \theta}(\boldsymbol{\sigma})$. Here, the variants of RNN architectures used in this paper are characterized by the \emph{autoregressive} structure of their probability distributions, which allows for a more efficient sampling compared to typical Markov-Chain Monte Carlo methods~\cite{germain2015mademaskedautoencoderdistribution, uria2016neuralautoregressivedistributionestimation,Wu_2019, Hibat-Allah2020-RNN, PhysRevLett.124.020503}:
\begin{align}
    P_{\bm{\theta}}(\boldsymbol{\sigma}) = \prod_{n=1}^N P_{\bm{\theta}}(\sigma_n|\boldsymbol{\sigma}_{<n}) \equiv  \prod_{n=1}^N p_n. \label{eq:p_n}
\end{align}
At step $n$, the probability of drawing a spin $\sigma_n$ depends on all previous spins denoted as $\boldsymbol{\sigma}_{<n}$. The corresponding phase is given by the sum of conditional phases~\cite{Hibat-Allah2020-RNN} as: \[\phi_{\bm{\theta}}(\boldsymbol{\sigma})=\sum_{n=1}^N  \phi(\sigma_n|\boldsymbol{\sigma}_{<n}).\]

% Introduction and overview of RNNs
\subsection{Recurrent Neural Network Wave Functions}

\begin{figure}[t] 
    \centering 
    \includegraphics[width=\linewidth]{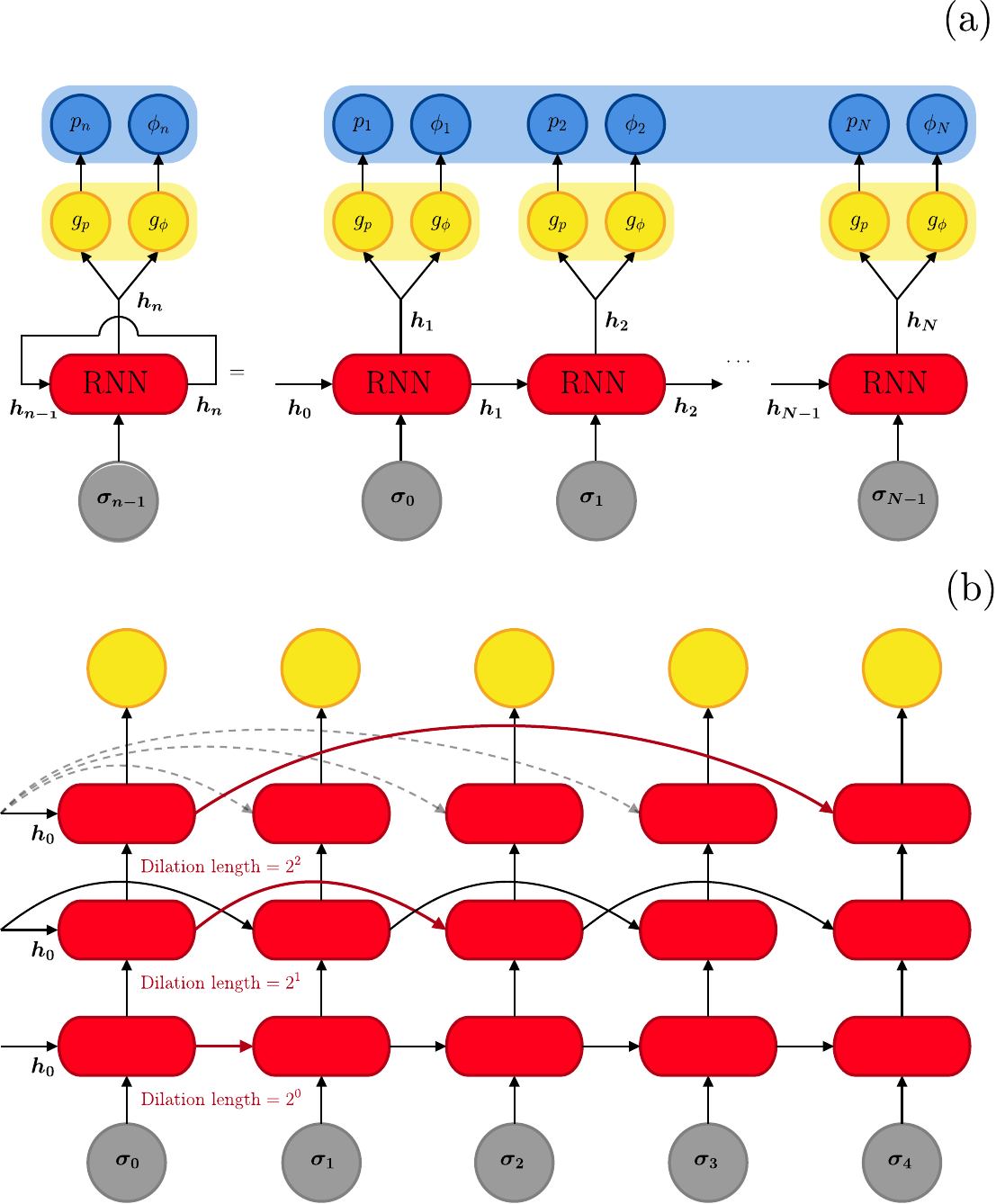}
    \caption{Schematic representation of the RNN architectures used. 
    (a) RNN wave function which sequentially generates the conditional probabilities $p_n$ and conditional phases $\phi_n$ by taking the one-hot encoded spin $\bm{\sigma}_{n-1}$ and hidden vector $\bm{h}_{n-1}$ as inputs at the $n^\text{th}$ step.
    (b) A dilated RNN illustration for a spin configuration with length $N = 5$. 
    For the $n^\text{th}$ site, the network uses the hidden vector $\bm{h}^{(l)}_{n-2^{l-1}}$ from the $(n-2^{l-1})^\text{th}$ site at the same layer $l$, in addition to $\bm{h}_n^{(l-1)}$ from the preceding layer. If the dilation length is larger than the spin position $n$, a zero initialization $\bm{h}_0$ is used instead.}
    \label{fig:RNN_schematic}
\end{figure}

RNNs have long been used in natural-language processing, notably for machine translation and speech recognition~\cite{cho2014learningphraserepresentationsusing, sak2014longshorttermmemorybased, sutskever2014sequencesequencelearningneural}. In addition, RNNs are theoretically well motivated because of their ability to act as universal approximators of Turing machines~\cite{Schafer2006-RNN-approximators}. At its core, an RNN operates by taking an initial input and augmenting it with a hidden vector $\bm{h}_{n}$, which acts as a memory unit that encodes correlations between the sequential inputs. The RNN then recursively generates hidden states for subsequent inputs, capturing key features of the sequence along the way. As described in Ref.~\cite{Hibat-Allah2020-RNN}, to construct a wave function with an RNN, we choose a spin configuration of length $N$ and sequentially feed the one-hot encoded spins $\bm{\sigma}_n$ as inputs. The task of the RNN is then to learn the conditional probabilities and the conditional phases, with which we can compute expectation values using the VMC framework. 

The simplest RNN wave function is based on a vanilla RNN cell~\cite{lipton2015criticalreviewrecurrentneural}. The recursion relation used to generate the hidden vector $\bm{h}_{n}$ for the $n^\mathrm{th}$ spin ${\sigma}_{n}$ of the spin configuration $\bm{\sigma}$ is given by:
\begin{align} 
    \bm{h}_{n} &= f\left(W 
    [\bm{h}_{n-1} ; \bm{\sigma}_{n-1}] + \bm{b}\right),\label{eq:RNN_hidden-state}
\end{align}
where $[. ;.]$ denotes the vector concatenation operation, $\bm{\sigma}_n \in \mathbb{R}^{2}$ is the one-hot encoding of $\sigma_n$, $\bm{b} \in \mathbb{R}^{d_h}$ is a bias, and $W \in \mathbb{R}^{d_h\times(d_h+2)}$ is a weight matrix. Here $d_h$ is the dimension of the hidden vector $\bm{h}_n$, which controls the expressivity of the RNN. The non-linear activation function $f$ is applied element-wise to the linearly transformed inputs. The sequential calculation of the hidden states is illustrated in Fig.~\ref{fig:RNN_schematic}(a).

The output layer, depicted by the shaded blocks in Fig.~\ref{fig:RNN_schematic}(a), generates the conditional probability $P_{\bm{\theta}}({\sigma}_n|\bm{\sigma}_{<n})$ and the conditional phase $\phi_{\bm{\theta}}(\sigma_n|\bm{\sigma}_{<n})$. Here the output layer uses the memory vector $\bm{h}_n$ as follows:
\begin{align}
    P_{\bm{\theta}}({\sigma}_n|\bm{\sigma}_{<n}) &= g_P(U\bm{h}_n+\bm{c})\cdot \bm{\sigma}_n \nonumber\\
    \phi_{\bm{\theta}}(\sigma_n|\bm{\sigma}_{<n}) &= g_\phi(V\bm{h}_n+\bm{d}) \cdot \bm{\sigma}_n,
    \label{eqn:output_layer}
\end{align}
where weights $U,V \in \mathbb{R}^{2 \times d_h}$, and biases $\bm{c},\bm{d} \in \mathbb{R}^2$. For this study, we choose the following element-wise non-linear activation functions $g_P$ and $g_\phi$, defined for a vector $\bm{v}$ as:
\begin{eqnarray}
        g_P(v_n) &= \mathrm{Softmax}(v_n) = \frac{\exp(v_n)}  {\sum_i \exp(v_i)} \nonumber \\
        g_\phi(v_n) &= \pi\ \mathrm{Softsign}(v_n) =  \pi\frac{v_n}{1+|v_n|},
    \label{eq:outputs}
\end{eqnarray}
which yield a normalized probability distribution $||P(\bm{\sigma})||_1=1$~\cite{Hibat-Allah2020-RNN}, and a conditional phase $\phi_n \in (-\pi,\pi)$. 
For the class of \emph{stoquastic} Hamiltonians with a sign-free structure in the computational basis~\cite{bravyi2015montecarlosimulationstoquastic}, the amplitudes are given by $\psi_{\bm{\theta}}(\boldsymbol{\sigma})=\sqrt{P_{\bm{\theta}}(\boldsymbol{\sigma})}$, where the Softsign output layer is not needed.

Together, the weights $\{W,U,V\}$ and biases $\{\bm{b,c,d}\}$ form the set of variational parameters $\{\bm{\theta}\}$ defined in Eq.~\eqref{eq:amp}. 
The same set of parameters $\{\bm{\theta}\}$ is shared across the sequence.
Throughout this paper, we use a Gated Recurrent Unit (GRU) cell~\cite{cho2014learningphraserepresentationsusing, zhou2016minimal} instead of the vanilla RNN cell described here. The details of the GRU cell implementation can be found in Appendix~\ref{app:Hyperparameters}. 

In addition to the challenges of expressivity and learnability of NQS ans\"atze, sequential architectures like RNNs suffer from the exploding/vanishing gradient problem. Among other strategies, introducing dilated connections into standard RNNs has been explored in the context of machine learning and combinatorial optimization~\cite{Chang2017DilatedRNN,Hibat-Allah2021-VNA,Khandoker_2023, Khandoker_2025}. The dilated connections are introduced through a deep multilayer RNN, with the length of dilated connections increasing with each layer, such that recurrent connections in higher layers introduce long-range interactions between spins. Related recurrent constructions have also been explored in other sequential settings, where stacked deep dynamic recurrent architectures were used to model complex multi-scale temporal dependencies~\cite{Aboussalah2020SDDRRL, koutnik2014clockworkrnn, chung2017hierarchicalmultiscalerecurrentneural}.

To incorporate the additional layers with dilated connections, we have to modify the recursion relation in Eq.~\eqref{eq:RNN_hidden-state} as follows: 
\begin{align}
\begin{split}
    \bm{h}_{n}^{(1)} &= f\left(W^{(1)} 
    [\bm{h}^{(1)}_{n-1} ; \bm{\sigma}_{n-1}] + \bm{b}^{(1)}\right) \label{eq:dilatedRNN_h} \\
    \bm{h}_{n}^{(l)} &= f\left(W^{(l)} 
    [\bm{h}^{(l)}_{\max(n-s^{(l)},0)} ; \bm{h}_{n}^{(l-1)}] + \bm{b}^{(l)}\right) \\ & \quad \text{for } 2 \leq l \leq L.
\end{split}
\end{align}
For this work, we set the dilation length $s^{(l)}=2^{l-1}$ for the $l^\mathrm{th}$ layer as demonstrated in Fig.~\ref{fig:RNN_schematic}(b). For a spin chain of length $N$, we can have a network depth of up to $L=\lceil\log_2(N)\rceil$. Finally, the output layer takes as input the hidden state $\bm{h}_n^{(L)}$ from the final deep layer $L$ to produce a final output which can be a conditional probability or a conditional phase as in Eq.~\eqref{eqn:output_layer}. Note that dilated RNNs admit a favorable scaling for a forward pass $\mathcal{O}(N \log N)$ compared with a quadratic $\mathcal{O}(N^2)$ scaling of standard Transformer models as a function of system size $N$~\cite{NIPS2017_3f5ee243}.

\subsection{Theoretical Scaling of Two-Point Correlations}

\begin{figure}
    \centering
    \includegraphics[width=\linewidth]{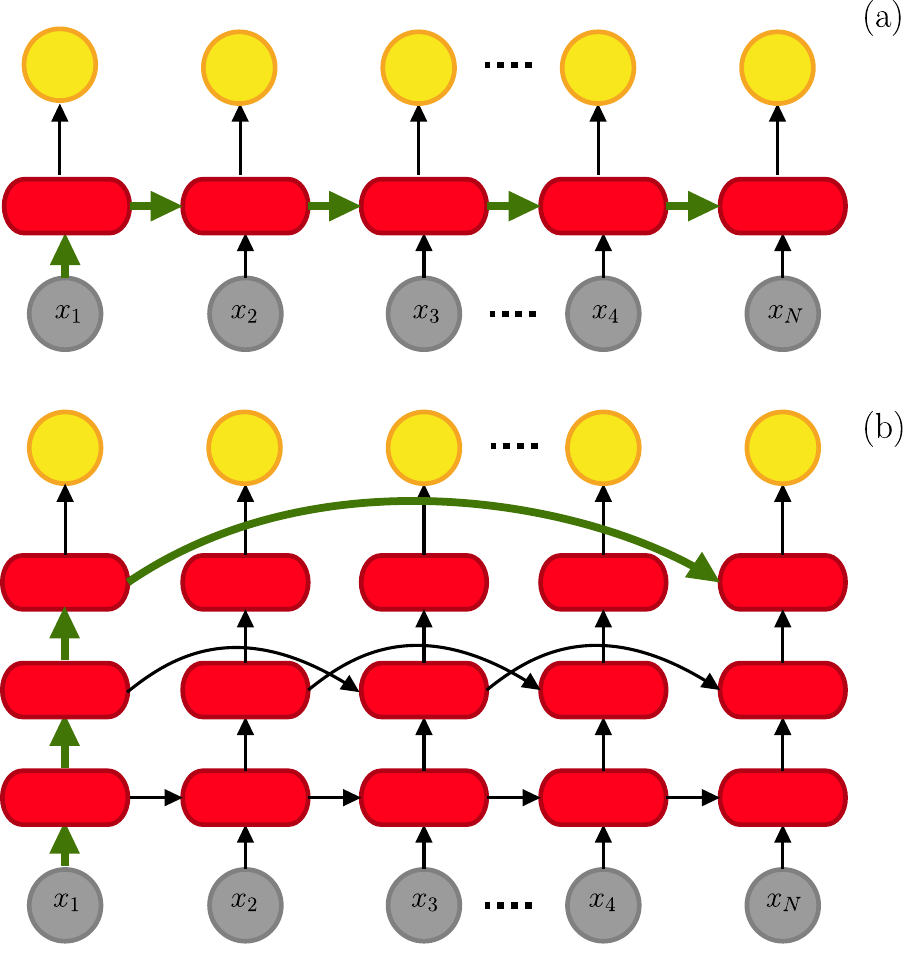}
    \caption{Illustration of the shortest path (in bold green arrows) between the last RNN cell at site $x_n$ and the first site $x_1$ of (a) the vanilla RNN and (b) the dilated RNN given a system size $N$. For the vanilla RNN, the shortest path between sites separated by lag $n$ is given by $\mathcal{O}(n)$, whereas for the dilated RNN, the shortest path has worst-case scaling $\mathcal{O}(\log(n))$.}
    \label{fig:geodesics}
\end{figure}

To gain insight into the theoretical scaling of two-point correlations in vanilla and dilated RNNs, we consider a linearized recurrent architecture with scalar input and output, and diagonal hidden dynamics. The vanilla RNN is simplified to:
\[
\bm{h}_n=\bm{w} x_{n-1}+U_h \bm{h}_{n-1},\qquad o_n=\bm{u}^\top \bm{h}_n + \bm{b},
\]
$\bm{u}, \bm{w} \in \mathbb{R}^{d_h}$ and $b \in \mathbb{R}$ are free parameters. The inputs are taken as scalar discrete variables $x_{i} = \pm 1$ for simplicity. We further assume that $U_h$ is a diagonalizable matrix, so without loss of generality, we take $U_h=\mathrm{diag}(\lambda_1,\dots,\lambda_{d_h})$ as a diagonal matrix. We impose a condition on the eigenvalues $ 0<\lambda_j<1$ to ensure that the RNN recursion is stable with no exponential growth~\cite{pascanu2013difficultytrainingrecurrentneural}. Furthermore, we make the simplifying assumption that the recursion activation function is the identity and that samples $x_n$ are generated from a discrete distribution as follows:
\[
P(x_n=\pm1\mid x_{<n})=\frac{e^{\pm o_n}}{2\cosh o_n}.
\]
In Appendix~\ref{app:proof}, we demonstrate that the first-order approximation of the connected two-point correlation function follows an exponential law in the distance lag \(n\),
\begin{equation*}
    C_{n}\ \propto\ \exp{(-c n)},
\end{equation*}
for a constant $c$ that depends on the parameters. We highlight that there are exceptions, namely $C_{n}$ remaining constant in the distance $1 \leq n \leq N$; however, this marginal case corresponds to a fine-tuned case as highlighted in Appendix~\ref{app:proof}. Note that this result is similar to the observation of Ref.~\cite{shen2019mutualinformationscalingexpressive} for an Elman RNN with continuous variables. This result can be understood by considering the longest-separation case shown in Fig.~\ref{fig:geodesics}(a), namely the path between the inputs $x_1$ and $x_N$ with distance lag $n = N-1$. In the vanilla architecture, the shortest path length scales linearly with the system size, $\ell_{\min}(n)\sim n$. As a result, any multiplicative attenuation of propagating signals $\sim \lambda^{\ell_{min}(n)}$ remains exponential in the lag $n$.

In contrast, for a dilated RNN, we use \(L\) layers with a dilation distance \(s^{(l)}=B^{\,l-1}\) ($B = 2$ in this work). Here, the hidden states at each layer, and the output are computed as follows:
\[
\bm{h}^{(1)}_n = \bm{w} x_{n-1} + U_h^{(1)} \bm{h}^{(1)}_{n-1},
\]
\[
\bm{h}^{(l)}_n = \bm{h}^{(l-1)}_{n} + U_h^{(l)} \bm{h}^{(l)}_{n-s^{(l)}}\ (2\le l\le L),
\]
\[
o_n=\bm{u}^\top \bm{h}^{(L)}_n + b,
\]
with \(U_h^{(l)}=\mathrm{diag}(\lambda^{(l)}_1,\dots,\lambda^{(l)}_{d_h})\) and \(0<\lambda^{(l)}_j<1\). For this simplified linearized dilated RNN, we show under explicit assumptions that the connected correlator \[C_n = \langle x_n x_1 \rangle - \langle x_n \rangle \langle x_1 \rangle,\]in the first-order perturbation, has a power-law lower bound with distance $n$ as
\[
C_{n}\ = \Omega \left(n^{-\alpha}\right).
\]
Here $\alpha$ is a positive exponent which depends on the parameter values. To illustrate this point, Fig.~\ref{fig:geodesics}(b) shows the longest-separation case between $x_1$ and $x_N$ with lag $n = N-1$, for which $\ell_{\min}(n)$ has worst-case scaling $\mathcal{O}(\log n)$ (see Corollary~\ref{cor:lmin_alogm} in Appendix~\ref{app:proof}). As a result, a multiplicative attenuation of the form $\lambda^{\ell_{\min}(n)}$ yields at least a power-law decay rather than an exponential decay. More details about our derivation are provided in Appendix~\ref{app:proof}. Interestingly, our result shares a similar spirit with the Multi-scale Renormalization Ansatz (MERA) in Tensor Networks, where going from Matrix Product States (MPS) to MERA transforms exponentially decaying correlations into long-range correlations decaying as a power-law~\cite{Vidal_2008}.

%%%%%%%%%%%%%%%%%% Results %%%%%%%%%%%%%%%%%%
\section{Results}

The results reported below are obtained using dilated RNN wave functions introduced in Sec.~\ref{sec:methods}. We focus our benchmark on the one-dimensional transverse-field Ising model and the one-dimensional Cluster state. Here, the variational parameters $\{\bm{\theta}\}$ are optimized by minimizing the energy expectation value $\bra{\Psi_{\bm{\theta}}} \hat{H} | \Psi_{\bm{\theta}} \rangle$ using Adam optimizer~\cite{kingma2017adammethodstochasticoptimization}.

Within a VMC optimization step, gradients of the energy with respect to the variational parameters are computed using automatic differentiation and used to iteratively update the RNN parameters until convergence of the ground-state energy is achieved~\cite{Hibat-Allah2020-RNN}. At the end of the optimization, we compute the ground-state energy and two-point correlation functions, which are presented and discussed in the following subsections. More details about our hyperparameters are found in Appendix~\ref {app:Hyperparameters}.

\subsection{1D Transverse-Field Ising Model}

\begin{figure}[t] 
    \centering 
        \includegraphics[width=\linewidth]{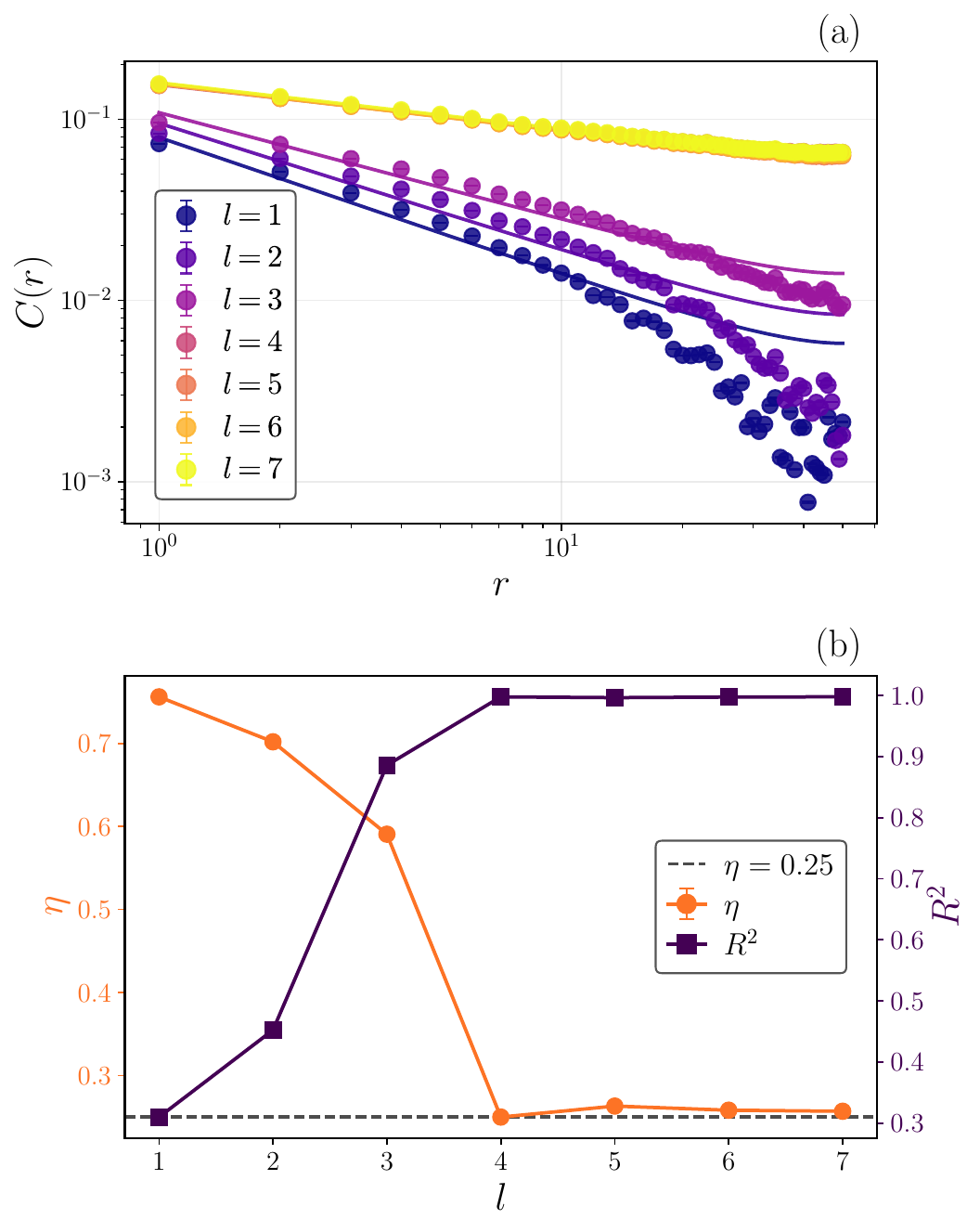}
        \label{fig:EvDilations_TFIM100}
    \caption{
        (a) A plot of the connected two-point correlation function $C(r)$ as a function of distance $r$ computed by dilated RNNs for different numbers of layers using the 1D periodic TFIM with $N = 100$ spins as a testbed. We observe that for $l \geq 4$ layers, the expected power-law behavior is captured by our dilated RNN ansatz. Error bars correspond to the one-standard deviation~\cite{becca_sorella_2017}. (b) A plot of the critical exponent $\eta$ versus different numbers of dilation layers obtained from fitting $C(r)$ against the chord length $L_r$ and the corresponding $R^2$ values. $\eta$ approaches the theoretical value of $0.25$ for $l \geq 4$ layers. The goodness of the fit, indicated by $R^2$, is close to unity for those layers. Error bars on $\eta$ denote the one-standard-error uncertainty from the covariance matrix of the least-squares fit.
    }
    \label{fig:TFIB}
\end{figure}

To demonstrate the advantage of the dilated architecture, we use the 1D Transverse-field Ferromagnetic Ising Model (TFIM) as a testbed for probing long-range correlations at the critical point~\cite{Mbeng_2024, Sachdev_2011}. Its Hamiltonian is given by:
\begin{equation}
    \hat{H}_{\text{TFIM}} = - \sum_{i = 1}^N \hat{\sigma}^{z}_i \hat{\sigma}^{z}_{i+1} - g \sum_{i=1}^N \hat{\sigma}^{x}_{i},
    \label{eq:TFIB}
\end{equation}
where $\hat{\sigma}^{(x,y,z)}_i$ are Pauli matrices acting on site $i$. Since this Hamiltonian is stoquastic~\cite{bravyi2015montecarlosimulationstoquastic}, we use a positive RNN wave function without a phase to target the ground state of this model~\cite{Hibat-Allah2020-RNN}.

To probe long-range correlations, we focus on the connected two-point correlation function $C(r)$ between pairs of spins $i,i+r$, defined as follows:
\begin{equation}
    C(r) = \braket{\hat{\sigma}^z_{i}\hat{\sigma}^z_{i+r}}-\braket{\hat{\sigma}^z_i}\braket{\hat{\sigma}^z_{i+r}}.
\end{equation}
In this paper, we choose $i = N/4$, and we limit the range of $r$ to $N/2$. Here, we focus on the critical point $g = 1$, and we adopt periodic boundary conditions (PBC), such that the two-point correlation in the ground state follows a power-law decay~\cite{di_francesco_conformal_1997}:
\begin{equation*}
    C(r) \propto L_r^{-\eta}.
\end{equation*}
Note that $\eta$ is a universal critical exponent that we extract in our fitting of the two-point correlations, in the log space, obtained at the end of training our RNN wave functions. To account for PBC, we use the chord length $L_r$, between the spins $i$ and $i+r$, defined as~\cite{Alba_2010,Pasquale_Calabrese_2004, PhysRevB.85.165121, HOLZHEY1994443, PhysRevLett.92.096402}
\begin{equation}
    L_r = \frac{N}{\pi} \sin \left (\frac{\pi r}{N}\right).
\end{equation}

The results are highlighted in Fig.~\ref{fig:TFIB}(a), where we show the evolution of $C(r)$ versus $r$ for different numbers of layers $l$ after training on 1D TFIM with $N = 100$ spins. For $l = 1$, we recover the traditional RNN wave function implementation with one layer, and for $l = 7$, we have a dilated RNN with the largest number of layers. As expected, the optimized RNN ansatz with $l = 1,2, 3$ layers failed to match a power-law behavior. This matches the empirical results of Ref.~\cite{Jreissaty_2026}, which provides numerical evidence of exponential-like decay in single-layer RNNs. For $l \geq 4$, we observe a strong agreement with a power-law decay with a coefficient of determination $R^2$ very close to $1$, confirming the goodness of our fit, as illustrated in Fig.~\ref{fig:TFIB}(b). Furthermore, the extracted exponent $\eta$ is in agreement with the expected $\eta = 0.25$ for a (1+1)-dimensional Ising CFT~\cite{di_francesco_conformal_1997}.

\subsection{The 1D Cluster State}

\begin{figure}[t] 
    \centering  \includegraphics[width=0.96\linewidth]{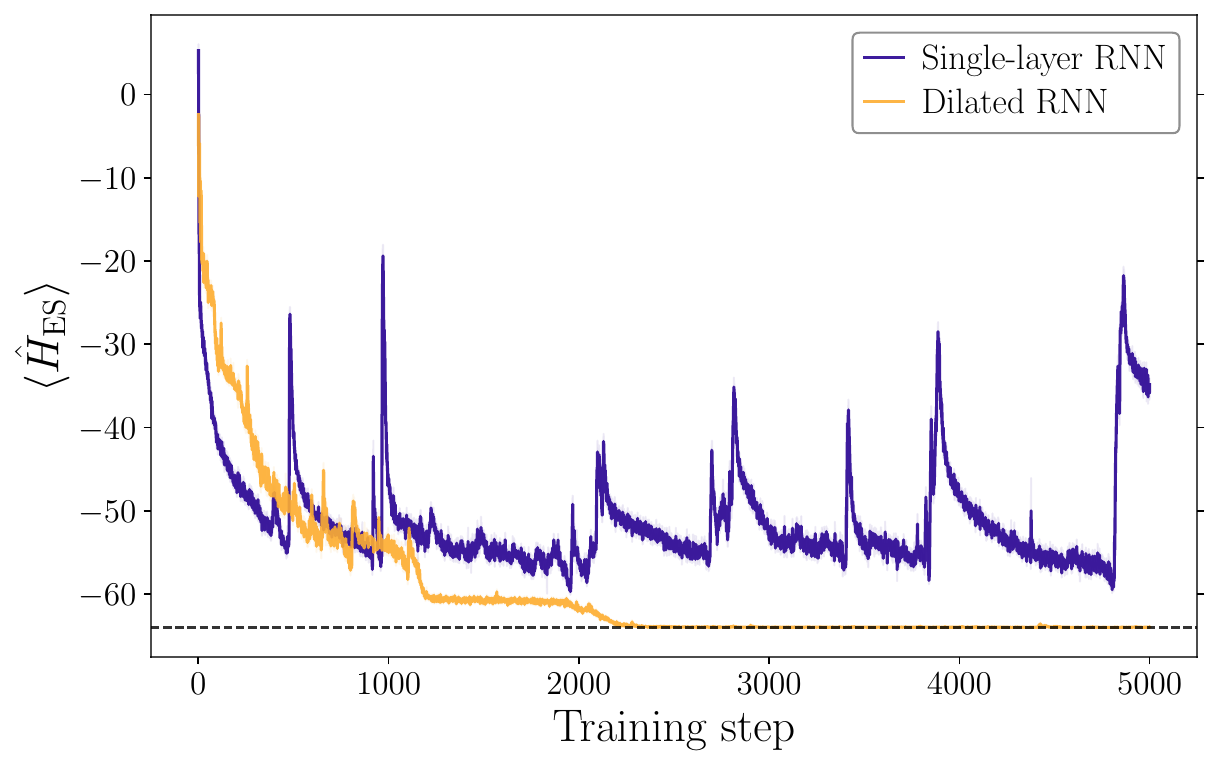}
        \label{fig:Cluter_state}
    \caption{RNN energy expectation values during training on the Cluster state benchmark for a single-layer RNN and a dilated RNN with $6$ layers. Note that the dilated architecture converges to the true ground state energy $E_G = -64$ indicated by the dashed line.
    }
    \label{fig:ClusterState}
\end{figure}

We next consider the 1D Cluster state, which corresponds to the ground state of the entanglement swapping Hamiltonian given by
\begin{align}
    \hat{H}_{\mathrm{ES}} = &- \sum _{k=2}^{N-2} \hat{\sigma}^x_{k-1}\hat{\sigma}^z_k \hat{\sigma}^x_{k+1}\nonumber\\
    &- \hat{\sigma}^z_1 \hat{\sigma}^x_2 - \hat{\sigma}^x_{N-1} \hat{\sigma}^x_N -  \hat{\sigma}^x_{N-2} \hat{\sigma}^z_{N-1} \hat{\sigma}^z_N.\label{eq:rotated_clusterH}
\end{align}
This state is highly entangled, and the entanglement is robust to spin measurements, making it an ideal candidate state for measurement-based quantum computing. Ref.~\cite{YangPreskill2024-ClusterStatePaper} shows that after applying a $y$-rotation $R_y(\theta)$ to $\hat{H}_{\rm ES}$, the conditional correlation $I(\theta)$ is upper bounded as:
\begin{equation*}
    I(\theta) \leq \cos(\theta)^{N-2}.
\end{equation*}
Additionally, Ref.~\cite{YangPreskill2024-ClusterStatePaper} reports performance degradation of variational RNN wave function simulations as $\theta\rightarrow0$ due to the presence of long-range conditional correlations, which is why we focus on studying the $\theta=0$ state as a benchmark for the dilated RNN. Since the ground state of this Hamiltonian is non-stoquastic, we use a complex dilated RNN wave function with a phase term. 

The results in Fig.~\ref{fig:ClusterState} demonstrate that the dilated RNN with $6$ layers can approximate the ground state energy for $N = 64$ within a relative error of $4(2) \times 10^{-5}$. This result outperforms our single-layer RNN and other single-layer RNN wave functions in Refs.~\cite{McNaughton2025-ku, YangPreskill2024-ClusterStatePaper}. It is also worth emphasizing that our single-layer RNN suffers from training instability, in contrast to the dilated RNN, where convergence to the ground state is smooth and stable. This observation supports the conclusion that our architectural design can not only enhance support for long-range conditional correlations but also improve trainability in a VMC training setup.

%%%%%%%%%%%%%%%%%% Conclusion %%%%%%%%%%%%%%%%%%
\section{Conclusions}

In this work, we introduced the dilated RNN architecture as a variational quantum state capable of encoding long-range correlations more effectively than standard RNN wave functions. Within a simplified linearized analysis, we showed that the first-order approximate connected two-point correlation function of a vanilla RNN decays exponentially in the asymptotic limit, whereas for a dilated RNN, it satisfies a power-law lower bound with distance under explicit assumptions. Beyond the NQS setting, our results support the broader view that incorporating suitable geometric structure into neural architectures can significantly improve the representation of long-range dependencies~\cite{Aboussalah2025GeoHNN}.

Our empirical results are consistent with the geometric scaling picture suggested by the linearized analysis. In particular, we show that for the 1D TFIM, we reproduce the expected power-law scaling in the two-point connected correlations of the ground state using a dilated RNN as opposed to an RNN with a small number of layers, which suffers from exponentially decaying correlations. We further reproduce the critical exponent $\eta$ of the (1+1)-dimensional Ising transition, confirming the success of our approach. Additionally, we use a complex dilated RNN wave function to target the one-dimensional Cluster state, which encodes long-range conditional correlations. Here, we show that a dilated RNN converges smoothly to the ground state energy, whereas a single-layer RNN exhibits unstable training behavior. Note that our empirical results are solely based on first-order gradient descent optimization, and we believe that applying second-order optimization techniques~\cite{Rende_2024, Chen2024, armegioiu2025functionalneuralwavefunctionoptimization} can further improve our results. We leave these investigations for future work.  

This work can be extended in multiple directions. In particular, a two-dimensional version of the dilated RNN can be devised and would be a great candidate for studying two-dimensional quantum many-body systems with long-range correlations. Additionally, we believe that a dilated RNN can likely handle logarithmic corrections to the area law of entanglement entropy~\cite{Eisert_2010} for critical quantum systems with long-range correlations. We expect future work incorporating dilated recurrent connections to make RNN wave functions a reliable and computationally cheaper alternative to simulating quantum many-body systems with long-range interactions, such as Rydberg atom arrays~\cite{Bernien2017}, and trapped ions~\cite{Zhang2017}.

%%%%%%%%%%%%%%%%% ACKNOWLEDGMENTS %%%%%%%%%%%%%%%%%%

\section*{Acknowledgments}
Computer simulations were made possible thanks to the Digital Research Alliance of Canada and the Math Faculty Computing Facility at the University of Waterloo. M.H. acknowledges support from the Natural Sciences and Engineering Research Council of Canada (NSERC) and the Digital Research Alliance of Canada. Research at Perimeter Institute is supported in part by the Government of Canada through the Department of Innovation, Science and Economic Development and by the Province of Ontario through the Ministry of Colleges and Universities.
%%%%%%%%%%%%%%%%% APPENDICES %%%%%%%%%%%%%%%%%%
\appendix

\section{Hyperparameters}

\label{app:Hyperparameters}
In this Appendix, we report the hyperparameters used in this paper. These hyperparameters are summarized in Tab.~\ref{tab:hyperparams}. Note that we use Adam optimizer~\cite{kingma2017adammethodstochasticoptimization} to minimize the energy expectation value in the VMC setup. Furthermore, we note that we use the Gated Recurrent Unit (GRU)~\cite{cho2014learningphraserepresentationsusing, zhou2016minimal} as our RNN cell following the standard implementation:
\begin{align}
\bm{g}_n &= \sigma\!\left( W_{\rm g}
\begin{bmatrix}
\bm{h}_{n-1}\\
\bm{\sigma}_{n-1}
\end{bmatrix}
+ \bm{b}_{\rm g} \right), \label{eq:gru_a}\\
\bm{r}_n &= \sigma\!\left( W_{\rm r}
\begin{bmatrix}
\bm{h}_{n-1}\\
\bm{\sigma}_{n-1}
\end{bmatrix}
+ \bm{b}_{\rm r} \right), \label{eq:gru_b}\\
\bm{h}^{\star}_n &= \tanh\!\left( \bm{r}_n \odot \bigl( W_{\rm h}  \bm{h}_{n-1} + \bm{b}_{\rm h}^{\star} \bigr)
+ W_{\rm in}\bm{\sigma}_{n-1} + \bm{b}_{\rm in}^{\star} \right), \label{eq:gru_c}\\
\bm{h}_n &= \bigl(\bm{1}-\bm{g}_n\bigr)\odot \bm{h}_{n-1} + \bm{g}_n \odot \bm{h}^{\star}_n.
\label{eq:gru_d}
\end{align}
Here, $\odot$ denotes the element-wise (Hadamard) product. The functions
$\sigma(\cdot)$ and $\tanh(\cdot)$ represent the standard sigmoid and hyperbolic
tangent activations, respectively. The reset gate $\bm{r}_n$ determines how
strongly the previous hidden state $\bm{h}_{n-1}$ influences the candidate state
$\bm{h}^{\star}_n$: small entries in $\bm{r}_n$ attenuate past information, while
entries close to one preserve it. The vector $\bm{g}_n$ acts as an update gate
that interpolates between retaining $\bm{h}_{n-1}$ and replacing it with the
candidate $\bm{h}^{\star}_n$. The matrices $W_{\rm g}, W_{\rm r}, W_{\rm h},
W_{\rm in}$ and bias terms $\bm{b}_{\rm g}, \bm{b}_{\rm r}, \bm{b}_{\rm h}^{\star}, \bm{b}_{\rm in}^{\star}$
are trainable parameters of the one-dimensional GRU cell.

\begin{table*}[p]
    \centering
    \begin{tabular}{|c|c|c|}\hline
       Figures & Hyperparameter & Value \\\hline
        \multirow{7}{*}{Fig.~\ref{fig:TFIB}} & Architecture & 1D positive RNN wave function\\
            & Number of samples used for training & $N_s = 100$ \\
            & Number of samples used for estimating correlations & $N_s = 100,000$ \\
            & Number of memory units & $d_h = 32$ \\
             & Training iterations & $100,000$  \\
             & Learning rate & $10^{-4}$\! \!  \\
             & Number of seeds & $1$  \\
             & Fitting window & $N/4$ to $3N/4$  \\
       \hline
               \multirow{5}{*}{Fig.~\ref{fig:ClusterState}} & Architecture & 1D complex RNN wave function\\
            & Number of samples used for training & $N_s = 100$ \\
            & Number of samples used for final energy estimation & $N_s = 50,000$ \\
            & Number of memory units & $d_h = 256$ \\
             & Learning rate & $10^{-3}$\! \! \\ 
             & Number of seeds & $1$  \\
             \hline
    \end{tabular}
    \caption{Hyperparameters used to obtain the results reported in this paper.}
    \label{tab:hyperparams}
\end{table*}

%%%%%%%%%%%%%%%%%%%%%%%%%%%%%%%%%%%%%%%%%%%%%%%%%%%%%%%%%%%%%%%%%%%%%%%%%%%%%%%%%%%%%

%%%%%%%%%%%%%%%%%%%%%%%%%%%%%%%%%%%%%%%%%%%%%%%%%%%%%%%%%%%%%%%%%%%%%%%%%%%%%%%%%%%%%
\section{Asymptotic decay of the two-point correlations for Vanilla and Dilated RNNs}
\label{app:proof}

In this Appendix, we provide a theoretical insight into the two-point correlation function scaling for simplified versions of the vanilla and dilated RNNs. We first discuss vanilla RNNs, and then we move to dilated RNNs. The next subsection collects the auxiliary results used in the derivation.

\subsection{Vanilla RNNs}
We assume that our inputs \(x_n\in\{-1,1\}\) and consider the linearized recurrence
\[
\bm{h}_n=\bm{w} x_{n-1}+U_h \bm{h}_{n-1},\qquad o_n=\bm{u}^\top \bm{h}_n+\bm{b},
\]
where \(\bm{h}_n\in\mathbb{R}^{d_h}\), \(\bm{w}, \bm{u}\in\mathbb{R}^{d_h}\), and \(U_h=\mathrm{diag}(\lambda_1,\ldots,\lambda_{d_h})\) with \(0<\lambda_j<1\). We assume \(0<\lambda_j<1\) so that each hidden mode is stable and attenuates monotonically with lag. More generally, the stability requirement is \(|\lambda_j|<1\), but restricting to positive eigenvalues avoids oscillatory behavior and keeps the asymptotic analysis transparent.

We assume \(\bm{h}_0=0\) and \(x_0=0\). The output channel is logistic:
\[
P(x_n=\pm1\mid x_{<n})=\frac{e^{\pm o_n}}{e^{o_n} + e^{-o_n}}=\frac{e^{\pm o_n}}{2\cosh o_n}.
\]
As a result, the conditional expectation is given as
\[
\begin{aligned}
\mathbb{E}[x_n \mid x_{<n}]
&= (+1)\,P(x_n=1 \mid x_{<n}), \\
&\quad + (-1)\,P(x_n=-1 \mid x_{<n}), \\
&= \frac{e^{o_n}}{2\cosh o_n} - \frac{e^{-o_n}}{2\cosh o_n}, \\
&= \frac{e^{o_n}-e^{-o_n}}{e^{o_n}+e^{-o_n}}, \\
&= \tanh o_n.
\end{aligned}
\]
We now study the connected correlator, which measures the dependence between site $1$ and site $n$
\[
C_n:=\mathbb{E}[x_nx_1]-\mathbb{E}[x_n]\mathbb{E}[x_1].
\]
Unrolling the recurrence gives
\[
o_n=b+\sum_{k<n}J_{n,k}\,x_k,\qquad J_{n,k}:=\bm{u}^\top U_h^{\,n-1-k}\bm{w}.
\]
Since \(J_{n,k}\), which is the weight measuring how strongly $x_k$ affects $o_n$, depends only on the lag \(m:=n-k\), we write \(J_{n,k}=k_m\) with
\[
k_m:=\bm{u}^\top U_h^{\,m-1}\bm{w}=\sum_{j=1}^{d_h}c_j\,\lambda_j^{m-1},\quad c_j:=u_jw_j,\quad m\ge1,
\]
meaning that $k_m$ is a sum of exponentially decaying modes, one for each hidden direction $j$. Since \(o_n=b+\delta_n\), where
\[
\delta_n:=\sum_{k<n}J_{n,k}x_k,
\]
we now quantify the regime in which \(\delta_n\) remains uniformly small so that a first-order expansion of \(\tanh(b+\delta_n)\) is justified. To make this weak-coupling expansion precise, define
\[
\varepsilon:=\sup_{n\ge1}\sum_{k<n}|J_{n,k}|
\le 
\sum_{j=1}^{d_h} |c_j| \sum_{m=1}^{n-1} \lambda_j^{m-1}
\le 
\sum_{j=1}^{d_h}\frac{|c_j|}{1-\lambda_j},
\]
and assume \(\varepsilon\ll1\). Since \(|x_k|=1\), the random field \(\delta_n\) satisfies \(|\delta_n|\le \varepsilon\) uniformly in \(n\). Therefore
\[
\tanh(b+\delta_n)=\tanh b+\beta\,\delta_n+r_n,\quad \beta:=1-\tanh^2 b,
\]
where \(r_n=O(\varepsilon^2)\) uniformly in \(n\). Multiplying by \(x_1\) and applying the tower rule, we get
\[
\begin{aligned}
\mathbb{E}[x_nx_1]
&=
\mathbb{E}\!\left[x_1\,\mathbb{E}[x_n\mid x_{<n}]\right], \\
&=
\tanh b\,\mathbb{E}[x_1]
+
\beta\sum_{k<n}J_{n,k}\,\mathbb{E}[x_kx_1]
+
O(\varepsilon^2).
\end{aligned}
\]
Similarly, by the tower property of conditional expectation, we have
\[
\begin{aligned}
\mathbb{E}[x_n]
&=
\mathbb{E}\!\left[\mathbb{E}[x_n\mid x_{<n}]\right] \\
&=
\mathbb{E}\!\left[ \tanh(b+\delta_n) \right] \\
&=
\tanh b
+
\beta\sum_{k<n}J_{n,k}\,\mathbb{E}[x_k]
+
O(\varepsilon^2).
\end{aligned}
\]
Hence, for \(n\ge2\),
\[
C_n
=
\beta\sum_{k=1}^{n-1}k_{n-k}\,C_k
+
O(\varepsilon^2).
\label{eq:corr_expan}
\]
Starting from the first-order correlation recursion
\begin{equation}
C_n
=
\beta\sum_{k=1}^{n-1}k_{n-k}\,C_k
+
R_n,
\qquad n\ge 2,
\label{eq:corr_expan_appendix_derivation}
\end{equation}
where \(R_n=O(\varepsilon^2)\), we now move on to generating functions. Let us define
\[
\begin{aligned}
C(z) &:=\sum_{n\ge1}C_nz^n,
\qquad
K(z):=\sum_{m\ge1}k_m z^m, \\
\qquad
R(z) &:=\sum_{n\ge2}R_n z^n.
\end{aligned}
\]
We multiply Eq.~\eqref{eq:corr_expan_appendix_derivation} by \(z^n\) and sum over \(n\ge2\). This gives
\[
\sum_{n\ge2} C_n z^n
=
\beta \sum_{n\ge2}\sum_{k=1}^{n-1} k_{n-k} C_k z^n
+
\sum_{n\ge2} R_n z^n.
\]
We now identify each term separately. First, since
\[
C(z)=\sum_{n\ge1}C_nz^n=C_1z+\sum_{n\ge2}C_nz^n,
\]
we have
\[
\sum_{n\ge2}C_nz^n = C(z)-C_1z.
\]
Next, consider the double sum:
\[
\sum_{n\ge2}\sum_{k=1}^{n-1} k_{n-k} C_k z^n.
\]
By setting
\[
m:=n-k,
\]
and since \(1\le k\le n-1\), we have \(m\ge1\), and therefore
\[
\sum_{n\ge2}\sum_{k=1}^{n-1} k_{n-k} C_k z^n
=
\sum_{k\ge1}\sum_{m\ge1} k_m\, C_k\, z^{k+m}.
\]
Factoring the powers of \(z\), this becomes
\[
\begin{aligned}
\sum_{k\ge1}\sum_{m\ge1} k_m\, C_k\, z^{k+m}
&=
\left(\sum_{k\ge1} C_k z^k\right)
\left(\sum_{m\ge1} k_m z^m\right) \\
&=
C(z)K(z).
\end{aligned}
\]
Finally,
\[
\sum_{n\ge2}R_n z^n = R(z).
\]
Substituting these identities into the summed equation yields
\[
C(z)-C_1z
=
\beta\,C(z)\,K(z)+R(z).
\]
If we neglect the remainder \(R(z)\), which is of order \(O(\varepsilon^2)\), we obtain the first-order approximation
\[
C^{\mathrm{app}}(z)-C_1z=\beta\,C^{\mathrm{app}}(z)\,K(z).
\]
Rearranging gives
\[
C^{\mathrm{app}}(z)\bigl(1-\beta K(z)\bigr)=C_1z,
\]
and hence the first-order approximation
\begin{equation}
    C^{\mathrm{app}}(z)=\frac{C_{1}z}{1-\beta K(z)}.
    \label{eq:first_order_approx}
\end{equation}
Using
\[
k_m=\sum_{j=1}^{d_h} c_j\,\lambda_j^{m-1},
\]
we have
\[
\begin{aligned}
K(z)
&=
\sum_{m\ge1} k_m z^m 
=
\sum_{j=1}^{d_h} c_j \sum_{m\ge1} \lambda_j^{m-1} z^m \\
&=
\sum_{j=1}^{d_h} c_j\, z \sum_{m\ge1} (\lambda_j z)^{m-1}.
\end{aligned}
\]
Since
\[
\sum_{m\ge1} (\lambda_j z)^{m-1}=\frac{1}{1-\lambda_j z},
\qquad |\lambda_j z|<1,
\]
it follows that
\[
K(z)=\sum_{j=1}^{d_h}\frac{c_j z}{1-\lambda_j z}.
\]
To extract the asymptotic behavior of the coefficients \(C_n\), we now study the singularities of the generating function
\[
C^{\mathrm{app}}(z)=\frac{C_1z}{1-\beta K(z)}.
\]
These singularities are determined by the zeros of the denominator, so we introduce
\[
D(z):=1-\beta K(z).
\]
Let \(\lambda_*:=\max_j\lambda_j\), and suppose that \(D\) has a smallest positive zero \(z_*\in(0,1/\lambda_*)\). Suppose further that \(z_*\) is a zero of multiplicity \(q\ge1\), so that
\[
D(z)=(z-z_*)^qG(z),\qquad G(z_*)\neq0.
\]
Using
\[
C^{\mathrm{app}}(z)=\frac{C_1z}{D(z)}
\qquad\text{and}\qquad
D(z)=(z-z_*)^qG(z),
\]
we get
\[
C^{\mathrm{app}}(z)=\frac{C_1}{(z-z_*)^q}\,\frac{z}{G(z)}.
\]
Since \(G(z_*)\neq0\), the factor \(z/G(z)\) is analytic at \(z_*\), so
\[
\frac{z}{G(z)}=\frac{z_*}{G(z_*)}+O(z-z_*).
\]
Substituting this into the previous expression gives
\[
C^{\mathrm{app}}(z)
=
\frac{C_1}{(z-z_*)^q}
\left(
\frac{z_*}{G(z_*)}+O(z-z_*)
\right).
\]
Therefore,
\[
C^{\mathrm{app}}(z)
=
\frac{C_1z_*}{G(z_*)}\,\frac{1}{(z-z_*)^q}
+
O\!\left(\frac{1}{(z-z_*)^{q-1}}\right).
\]
Since
\[
C^{\mathrm{app}}(z)=\sum_{n\ge1} C_n^{(\mathrm{app})} z^n,
\]
we denote by \([z^n]\,F(z)\) the coefficient of \(z^n\) in the power-series expansion of \(F(z)\). In particular,
\[
[z^n]\,C^{\mathrm{app}}(z)=C_n^{(\mathrm{app})}.
\]
To extract the large-\(n\) behavior of the coefficients \(C_n^{(\mathrm{app})}\), we now use the local form of \(C^{\mathrm{app}}(z)\) near its dominant singularity. Suppose that \(z_*\) is the dominant singularity of \(C^{\mathrm{app}}\), and that it is a zero of \(D(z)\) of multiplicity \(q\ge 1\). From the previous local expansion, we have
\[
C^{\mathrm{app}}(z)
=
\frac{C_1z_*}{G(z_*)}\,\frac{1}{(z-z_*)^q}
+
O\!\left(\frac{1}{(z-z_*)^{q-1}}\right).
\]
It is convenient to rewrite the leading singular term in the standard transfer-theorem form. Since
\[
z-z_*=-z_*\left(1-\frac{z}{z_*}\right),
\]
we obtain
\[
\frac{1}{(z-z_*)^q}
=
\frac{(-1)^q}{z_*^q}\left(1-\frac{z}{z_*}\right)^{-q}.
\]
Hence
\[
C^{\mathrm{app}}(z)
=
B\left(1-\frac{z}{z_*}\right)^{-q}
+
O\!\left(\left(1-\frac{z}{z_*}\right)^{-(q-1)}\right),
\]
where
\[
B:=\frac{C_1z_*}{G(z_*)}\,\frac{(-1)^q}{z_*^q}.
\]
We now use the standard coefficient formula
\[
(1-w)^{-q}
=
\sum_{n\ge 0}\binom{n+q-1}{q-1}w^n,
\qquad q\in\mathbb{N}.
\]
Taking \(w=z/z_*\), this gives
\[
[z^n]\left(1-\frac{z}{z_*}\right)^{-q}
=
z_*^{-n}\binom{n+q-1}{q-1}.
\]
Moreover, as \(n\to\infty\),
\[
\binom{n+q-1}{q-1}
\sim
\frac{n^{q-1}}{(q-1)!}.
\]
Therefore singularity analysis~\cite{flajolet2009analytic} yields
\[
C_n^{(\mathrm{app})}
\sim
A\,n^{q-1}z_*^{-n},
\]
for some constant \(A\neq 0\). Equivalently, writing
\[
\rho:=z_*^{-1},
\]
we obtain
\[
C_n^{(\mathrm{app})}\sim A\,n^{q-1}\rho^n.
\]
This formula makes the exponential character completely explicit: the factor \(n^{q-1}\) is only polynomial, whereas the term \(\rho^n=z_*^{-n}\) is exponential in \(n\). In particular, if \(|z_*|>1\), then \(|\rho|<1\), so $C_n^{(\mathrm{app})}$ decays exponentially fast, up to the polynomial prefactor. Thus, whenever the dominant singularity lies strictly outside the unit circle, the first-order approximation predicts exponentially decaying correlations.

By contrast, if \(|z_*|<1\), then \(|\rho|>1\), and the approximation predicts exponential growth. Such a behavior cannot describe the exact correlator, because for \(x_n\in\{-1,1\}\) one has
\begin{equation}
\label{eq:bounded_correlator}
|C_n|
=
\bigl|\mathbb{E}[x_nx_1]-\mathbb{E}[x_n]\mathbb{E}[x_1]\bigr|
\le 2,
\end{equation}
Thus
\[
\sum_{n\ge1} |C_n z^n|
\le
2\sum_{n\ge1} |z|^n.
\]
As a result, the geometric series on the right converges whenever \(|z|<1\). Therefore, the exact generating function
\[
C(z)=\sum_{n\ge1} C_n z^n
\]
converges absolutely for every \(|z|<1\), which implies that its radius of convergence is at least \(1\), $i.e.$ $C(z)$ has no singularities for \(|z|<1\). Therefore a dominant singularity strictly inside the unit disk cannot govern the exact correlation sequence; it only indicates a breakdown of the first-order truncation.

The boundary case \(|z_*|=1\), and in particular \(z_*=1\), is different. In that case \(\rho=1\), so there is no exponential factor, and the first-order approximation gives
\[
C_n^{(\mathrm{app})}\sim A\,n^{q-1}.
\]
Thus, if \(q=1\), the approximation predicts a non-decaying \(O(1)\) contribution. Otherwise, when \(q\ge 2\), it predicts polynomial growth, which again cannot describe the exact bounded two-point connected correlations~\eqref{eq:bounded_correlator}. Hence the case \(|z_*|=1\) should be interpreted as a marginal or critical threshold of the first-order approximation. Furthermore, the case $|z_*| = 1$ is a fine-tuned case with zero measure, which is unlikely to happen when randomly initializing or training the parameters of the vanilla RNN.

A sufficient condition ensuring that the first-order denominator has no zeros on the closed unit disk is
\[
\sup_{|z|\le 1} |\beta K(z)|<1.
\]
Indeed, if this holds, then for every \(|z|\le 1\),
\[
|D(z)-1|=|\beta K(z)|<1,
\]
so in particular \(D(z)\neq 0\) on \(|z|\le 1\). Consequently, all zeros of \(D(z)\) lie strictly outside the unit disk, and the dominant singularity \(z_*\) of \(C^{\mathrm{app}}\) satisfies \(|z_*|>1\). In that regime, the first-order approximation predicts exponential decay of correlations, consistent with the finite-memory behavior expected for vanilla RNNs~\cite{shen2019mutualinformationscalingexpressive}.

\subsection{Dilated RNNs}
\label{app:drnn_proof}

We now consider a family of linearized dilated models indexed by the system size \(N\), with depth
\[
L=\lceil\log_B N\rceil
\]
and dilation schedule
\[
s^{(l)}=B^{l-1},\qquad 1\le l\le L.
\]
Here \(\lceil \log_B N\rceil\) denotes the ceiling of \(\log_B N\), i.e., the smallest integer greater than or equal to \(\log_B N\). For analytic tractability, we restrict our analysis to the diagonal recurrent map
\[
U_h=\mathrm{diag}(\lambda_1,\dots,\lambda_{d_h}),\qquad 0<\lambda_j<1,
\]
shared across all layers. The recursion relations for each layer are given as:
\[
\bm{h}_n^{(1)}=\bm{w} x_{n-1}+U_h \bm{h}_{n-1}^{(1)},
\]
\begin{equation}
    \bm{h}_n^{(l)}=\bm{h}_n^{(l-1)}+U_h \bm{h}_{n-s^{(l)}}^{(l)},\qquad 2\le l\le L.
    \label{eq:dilatedlayers}
\end{equation}
Additionally, the output is as follows:
\[
o_n=\bm{u}^\top \bm{h}_n^{(L)}+b.
\]
Similar to the vanilla case, one can unroll the recursion and write \(o_n\) as a linear combination of past inputs:
\[
o_n=b+\sum_{m\ge1} k_m\,x_{n-m},
\]
where \(k_m\) is the total effective weight of all computational paths that transmit the input \(x_{n-m}\) to the output \(o_n\). To define this precisely, let \(\mathcal P_m\) denote the set of all admissible recurrent paths in the dilated computation graph that connect the input at lag \(m\) to the output, i.e., all paths whose total time shift is exactly \(m\). For each path \(p\in\mathcal P_m\), let \(\ell(p)\) be the number of horizontal recurrent edges along \(p\). Since each horizontal recurrent edge contributes one factor of \(U_h\), a path of length \(\ell(p)\) contributes a factor \(U_h^{\ell(p)}\). Summing over all such paths gives
\begin{equation*}
k_m
=
\bm{u}^\top\!\Big(\sum_{p\in\mathcal P_m}U_h^{\,\ell(p)}\Big)\bm{w}.
\end{equation*}
Because \(U_h\) is diagonal, this can be written componentwise as
\begin{equation}
k_m
=
\sum_{j=1}^{d_h} c_j \sum_{p\in\mathcal P_m}\lambda_j^{\,\ell(p)},
\qquad c_j:=u_jw_j.
\label{eq:lambdaweight}
\end{equation}
Thus, in the dilated model, the lag kernel \(k_m\) is no longer a single exponential in \(m\); instead, it is a sum over all admissible paths whose weights depend on their recurrent length.

To understand the contribution of the shortest path, we first determine the smallest possible number of recurrent edges needed to realize a total lag \(m\). We denote this minimal number by \(\ell_{\min}(m)\). Write \(m\) in base \(B\) as
\[
m=\sum_{r=0}^{R} b_r B^r,
\qquad
b_r\in\{0,1,\dots,B-1\}.
\]
We then define the base-\(B\) digit sum of \(m\) by
\[
s_B(m):=\sum_{r=0}^{R} b_r.
\]
Lemma~\ref{lem:lmin_log_formal} shows that
\[
\ell_{\min}(m)=s_B(m).
\]
In other words, the smallest number of recurrent edges needed to realize a lag \(m\) is obtained by adding the digits of \(m\) in base \(B\). Corollary~\ref{cor:lmin_alogm} then gives the worst-case bound
\[
\ell_{\min}(m)\le (B-1)\bigl(\lfloor \log_B m\rfloor+1\bigr),
\]
which shows that \(\ell_{\min}(m)\) grows at most logarithmically with \(m\).

Furthermore, for each hidden mode \(j\),
\[
\sum_{p\in\mathcal P_m}\lambda_j^{\ell(p)}
\geq 
\lambda_j^{\ell_{\min}(m)}.
\]
Thus the attenuation at lag \(m\) is lower-bounded by \(\lambda_j^{\ell_{\min}(m)}\). Additionally, as \(0<\lambda_j<1\), the map \(x\mapsto \lambda_j^x\) is decreasing. Therefore the upper bound on \(\ell_{\min}(m)\) yields the lower bound
\[
\lambda_j^{\ell_{\min}(m)}
\ge
\lambda_j^{(B-1)(\lfloor \log_B m\rfloor+1)}.
\]
Using
\[
\lfloor \log_B m\rfloor+1\le \log_B m+1,
\]
we further obtain
\[
\lambda_j^{\ell_{\min}(m)}
\ge
\lambda_j^{(B-1)(\log_B m+1)}
=
\lambda_j^{B-1}\,\lambda_j^{(B-1)\log_B m}.
\]
Rewriting the logarithmic exponent as a power-law gives
\[
\lambda_j^{(B-1)\log_B m}
=
m^{(B-1)\log_B\lambda_j}.
\]
Since \(\log_B\lambda_j<0\), defining
\[
\alpha_j:=-(B-1)\log_B\lambda_j>0
\]
yields
\[
\lambda_j^{\ell_{\min}(m)}
\ge
\lambda_j^{B-1}\,m^{-\alpha_j}.
\]
Equivalently,
\[
\lambda_j^{\ell_{\min}(m)}=\Omega(m^{-\alpha_j}).
\]
As a result, we have
\[
\sum_{p\in\mathcal P_m}\lambda_j^{\ell(p)}
=
\Omega(m^{-\alpha_j}).
\]
At this point, to obtain a lower bound on the full kernel \(k_m\), we impose an additional sign-alignment assumption on the couplings:
\[
c_j=u_jw_j\ge0 \qquad \text{for all } j,
\]
and
\[
c_{j_*}>0
\qquad \text{for at least one mode } j_*.
\]
This is a sufficient assumption used only for the lower-bound argument below; it excludes cancellations between hidden modes in the read-in/read-out couplings. Under this assumption,
\begin{align*}
  k_m
&=
\sum_{j=1}^{d_h} c_j \sum_{p\in\mathcal P_m}\lambda_j^{\ell(p)}
\ge
c_{j_*}\sum_{p\in\mathcal P_m}\lambda_{j_*}^{\ell(p)}.
\end{align*}
Hence the lag kernel is bounded from below by a power-law:
\[
k_m=\Omega(m^{-\alpha_{j_*}}).
\]

We now transfer this lower bound to the first-order approximate two-point correlation. Under the same first-order weak-coupling approximation used in the vanilla case, the same formal derivation yields Eq.~\eqref{eq:first_order_approx} for the dilated model:
\[
C^{\mathrm{app}}(z)=\sum_{n\ge1} C_n^{\mathrm{app}} z^n
=
\frac{C_1z}{1-\beta K(z)},
\quad
K(z)=\sum_{m\ge1} k_m z^m.
\]
Here
\[
\beta=1-\tanh^2 b>0,
\]
and, since \(x_1^2=1\),
\[
C_1
=
\mathbb{E}[x_1^2]-\mathbb{E}[x_1]^2
=
1-\tanh^2 b
=
\beta
>
0.
\]
Moreover, under the sign-alignment assumption above, we have \(k_m\ge0\) for all \(m\). Therefore, the geometric expansion may be interpreted coefficientwise as a formal power series:
\[
C^{\mathrm{app}}(z)
=
C_1z\sum_{r\ge0}\beta^r K(z)^r,
\]
All coefficients in this expansion are nonnegative. In particular, retaining only the \(r=1\) term yields the coefficientwise lower bound
\[
C_n^{\mathrm{app}}
\ge
C_1\beta\,k_{n-1},
\]
for $n \geq 2$. Combining this with the lower bound on \(k_m\) gives
\[
C_n^{\mathrm{app}}
=
\Omega(n^{-\alpha_{j_*}}).
\]
We therefore conclude that, under the additional sign-alignment condition on the couplings, the first-order approximate connected correlation is bounded from below by a power-law. In particular, within this regime, it cannot decay faster than a power-law. This is the sense in which dilation changes the effective correlation behavior of the autoregressive ansatz: by reducing the minimal path length from linear to logarithmic in the lag, it replaces the fixed exponential attenuation of the vanilla architecture by power-law-type decay. We stress, however, that this is a lower-bound statement for the first-order approximation of the two-point correlation functions.

\subsection{Auxiliary Lemmas}

In this subsection, we show the intermediate results needed in subsection~\ref{app:drnn_proof}. Before stating the lemmas, we specify the admissible path structure used throughout this subsection. An admissible path is a directed path in the layered computation graph that uses horizontal recurrent edges and vertical inter-layer edges, where vertical moves are only upward. We define the path length $\ell(p)$ to be the number of horizontal recurrent edges along $p$; vertical edges do not contribute to $\ell(p)$ (see Eq.~\eqref{eq:dilatedlayers}) and therefore do not affect the $\lambda$-weight of each path as in Eq.~\eqref{eq:lambdaweight}.

\begin{lemma}[Minimal path length formula and uniqueness]
\label{lem:lmin_log_formal}
Fix an integer $B\ge 2$ and consider the exponential dilation schedule
\[
s^{(l)}=B^{l-1},\qquad l\ge 1.
\]
For a lag $m\ge 1$, let $\ell_{\min}(m)$ denote the minimal number of
horizontal recurrent edges needed to realize a total distance shift of exactly $m$, where one is
allowed to use edges of lengths $\{1,B,B^2,\dots\}$ any nonnegative integer number
of times. Let us write the unique base-$B$ expansion 
\[
\begin{aligned}
m&=\sum_{r=0}^{R} b_r B^r,\\
b_r&\in\{0,1,\dots,B-1\},\qquad
b_R\neq 0,\qquad
R=\lfloor \log_B m\rfloor.
\end{aligned}
\]
Then the minimal number of edges is exactly the base-$B$ digit sum:
\[
\ell_{\min}(m)=\sum_{r=0}^{R} b_r.
\]
In particular,
\[
1\le \ell_{\min}(m)\le (B-1)\bigl(\lfloor \log_B m\rfloor+1\bigr),
\]
and therefore
\[
\ell_{\min}(m)=O(\log m).
\]
For the binary schedule $B=2$, this reduces to
\[
\ell_{\min}(m)=\sum_{r=0}^{R}\mathbf{1}\{b_r\neq 0\},
\]
since $b_r\in\{0,1\}$ in that case.
\end{lemma}

\begin{proof}
Any directed path realizing displacement \(m\) determines a representation
\[
m=\sum_{k=0}^{\lfloor \log_B m\rfloor} b_k B^k,
\qquad
b_k\in\mathbb N,
\]
where \(b_k\) is the number of horizontal edges of length \(B^k\). The number of horizontal edges in the path is therefore
\[
q=\sum_{k=0}^{\lfloor \log_B m\rfloor} b_k.
\]
Suppose that for some \(k\) one has \(b_k\ge B\). Then \(B\) jumps of size \(B^k\) may be replaced by one jump of size \(B^{k+1}\), since
\[
B\cdot B^k = B^{k+1}.
\]
This preserves the total displacement \(m\), while reducing the number of horizontal edges by \(B-1\). Hence, no shortest path can have \(b_k\ge B\). It follows that, for a shortest path, all coefficients satisfy
\[
0\le b_k \le B-1.
\]
Therefore, the coefficients \(b_k\) must coincide with the digits of the base-\(B\) expansion of \(m\), which is unique. Thus, once the counts $b_k$ of horizontal edges at each layer are fixed, the admissible path is uniquely determined, as vertical moves between layers are monotone upward and do not contribute to horizontal displacements.

Since each digit satisfies $0\le b_r\le B-1$ and there are exactly
$R+1=\lfloor \log_B m\rfloor+1$ digits in the expansion, we obtain
\[
\begin{aligned}
1\le \ell_{\min}(m)
&=\sum_{r=0}^{R} b_r \\
&\le (B-1)(R+1) \\
&= (B-1)\bigl(\lfloor \log_B m\rfloor+1\bigr).
\end{aligned}
\]
This proves $\ell_{\min}(m)=O(\log m)$.
For $B=2$, each digit belongs to $\{0,1\}$, so
\[
\sum_{r=0}^{R} b_r
=
\sum_{r=0}^{R}\mathbf{1}\{b_r\neq 0\},
\]
which gives the binary special case.
\end{proof}

\begin{corollary}[Worst-case and average logarithmic scaling]
\label{cor:lmin_alogm}
Under the assumptions of Lemma~\ref{lem:lmin_log_formal}, define
\[
s_B(0):=0.
\]
For $m\ge 1$, write the unique base-$B$ expansion of $m$ as
\[
m=\sum_{r=0}^{\lfloor \log_B m\rfloor} b_r B^r,
\qquad b_r\in\{0,1,\dots,B-1\},
\]
and define
\[
s_B(m):=\sum_{r=0}^{\lfloor \log_B m\rfloor} b_r.
\]
Also set $\ell_{\min}(0):=0$. Then
\[
\ell_{\min}(m)=s_B(m)
\qquad (m\ge 0).
\]

Consequently:
\begin{enumerate}
    \item[\rm (i)] \textit{Worst case over a full base-$B$ box:}
    for every integer $R\ge 0$,
    \[
    \max_{0\le m\le B^{R+1}-1}\ell_{\min}(m)
    =
    (B-1)(R+1),
    \]
    with equality at $m=B^{R+1}-1$.
    In particular, for every $m\ge 1$,
    \[
    \ell_{\min}(m)\le (B-1)\bigl(\lfloor \log_B m\rfloor+1\bigr),
    \]
    and hence
    \[
    \ell_{\min}(m)\le \frac{B-1}{\log B}\,\log m + O(1).
    \]

    \item[\rm (ii)] \textit{Average case over a full base-$B$ box:}
    if
    \[
    m \sim \mathrm{Unif}\{0,1,\dots,B^{R+1}-1\},
    \]
    then
    \[
    \mathbb{E}\bigl[\ell_{\min}(m)\bigr]
    =
    \frac{B-1}{2}(R+1).
    \]
    Equivalently, if $N=B^{R+1}$, then
    \[
    \mathbb{E}\bigl[\ell_{\min}(m)\bigr]
    =
    \frac{B-1}{2\log B}\,\log N.
    \]
\end{enumerate}
\end{corollary}

\begin{proof}
For $m\ge 1$, the identity $\ell_{\min}(m)=s_B(m)$ is exactly the statement of
Lemma~\ref{lem:lmin_log_formal}; for $m=0$ it holds by the convention
$\ell_{\min}(0)=s_B(0)=0$.

For part (i), every integer $m\in\{0,1,\dots,B^{R+1}-1\}$ has a base-$B$
expansion
\[
m=\sum_{r=0}^{R} b_r B^r,
\qquad b_r\in\{0,1,\dots,B-1\}.
\]
Hence
\[
\ell_{\min}(m)=\sum_{r=0}^{R} b_r \le (B-1)(R+1).
\]
This upper bound is attained at
\[
m=B^{R+1}-1=\sum_{r=0}^{R} (B-1)B^r,
\]
for which all digits equal $B-1$. Therefore
\[
\max_{0\le m\le B^{R+1}-1}\ell_{\min}(m)=(B-1)(R+1).
\]
Now take $R=\lfloor \log_B m\rfloor$ to obtain
\[
\ell_{\min}(m)\le (B-1)\bigl(\lfloor \log_B m\rfloor+1\bigr).
\]
Since
\[
\lfloor \log_B m\rfloor+1
=
\frac{\log m}{\log B}+O(1).
\]
It follows that
\[
\ell_{\min}(m)\le \frac{B-1}{\log B}\,\log m + O(1).
\]

For part (ii), if $m$ is uniformly distributed on $\{0,1,\dots,B^{R+1}-1\}$,
then its digits $b_0,\dots,b_R$ are independent and each uniformly distributed on
$\{0,1,\dots,B-1\}$. Therefore
\[
\mathbb{E}[b_r]
=
\frac{0+1+\cdots+(B-1)}{B}
=
\frac{B-1}{2}.
\]
By linearity of expectation,
\[
\mathbb{E}\bigl[\ell_{\min}(m)\bigr]
=
\mathbb{E}\Bigl[\sum_{r=0}^{R} b_r\Bigr]
=
\sum_{r=0}^{R} \mathbb{E}[b_r]
=
(R+1)\frac{B-1}{2}.
\]
If we set $N=B^{R+1}$, then $R+1=\log_B N = \dfrac{\log N}{\log B}$, hence
\[
\mathbb{E}\bigl[\ell_{\min}(m)\bigr]
=
\frac{B-1}{2\log B}\,\log N.
\]
This proves the claim.
\end{proof}

%%%%%%%%%%%%%%%%%% REFERENCES %%%%%%%%%%%%%%%%%%
\clearpage
\bibliography{Biblio}

\begin{thebibliography}{60}%
\makeatletter
\providecommand \@ifxundefined [1]{%
 \@ifx{#1\undefined}
}%
\providecommand \@ifnum [1]{%
 \ifnum #1\expandafter \@firstoftwo
 \else \expandafter \@secondoftwo
 \fi
}%
\providecommand \@ifx [1]{%
 \ifx #1\expandafter \@firstoftwo
 \else \expandafter \@secondoftwo
 \fi
}%
\providecommand \natexlab [1]{#1}%
\providecommand \enquote  [1]{``#1''}%
\providecommand \bibnamefont  [1]{#1}%
\providecommand \bibfnamefont [1]{#1}%
\providecommand \citenamefont [1]{#1}%
\providecommand \href@noop [0]{\@secondoftwo}%
\providecommand \href [0]{\begingroup \@sanitize@url \@href}%
\providecommand \@href[1]{\@@startlink{#1}\@@href}%
\providecommand \@@href[1]{\endgroup#1\@@endlink}%
\providecommand \@sanitize@url [0]{\catcode `\\12\catcode `\$12\catcode `\&12\catcode `\#12\catcode `\^12\catcode `\_12\catcode `\%12\relax}%
\providecommand \@@startlink[1]{}%
\providecommand \@@endlink[0]{}%
\providecommand \url  [0]{\begingroup\@sanitize@url \@url }%
\providecommand \@url [1]{\endgroup\@href {#1}{\urlprefix }}%
\providecommand \urlprefix  [0]{URL }%
\providecommand \Eprint [0]{\href }%
\providecommand \doibase [0]{http://dx.doi.org/}%
\providecommand \selectlanguage [0]{\@gobble}%
\providecommand \bibinfo  [0]{\@secondoftwo}%
\providecommand \bibfield  [0]{\@secondoftwo}%
\providecommand \translation [1]{[#1]}%
\providecommand \BibitemOpen [0]{}%
\providecommand \bibitemStop [0]{}%
\providecommand \bibitemNoStop [0]{.\EOS\space}%
\providecommand \EOS [0]{\spacefactor3000\relax}%
\providecommand \BibitemShut  [1]{\csname bibitem#1\endcsname}%
\let\auto@bib@innerbib\@empty
%</preamble>
\bibitem [{\citenamefont {Androsiuk}\ \emph {et~al.}(1993)\citenamefont {Androsiuk}, \citenamefont {Kułak},\ and\ \citenamefont {Sienicki}}]{ANDROSIUK1993377}%
  \BibitemOpen
  \bibfield  {author} {\bibinfo {author} {\bibfnamefont {J.}~\bibnamefont {Androsiuk}}, \bibinfo {author} {\bibfnamefont {L.}~\bibnamefont {Kułak}}, \ and\ \bibinfo {author} {\bibfnamefont {K.}~\bibnamefont {Sienicki}},\ }\bibfield  {title} {\enquote {\bibinfo {title} {Neural network solution of the schrödinger equation for a two-dimensional harmonic oscillator},}\ }\href {\doibase https://doi.org/10.1016/0301-0104(93)80153-Z} {\bibfield  {journal} {\bibinfo  {journal} {Chemical Physics}\ }\textbf {\bibinfo {volume} {173}},\ \bibinfo {pages} {377--383} (\bibinfo {year} {1993})}\BibitemShut {NoStop}%
\bibitem [{\citenamefont {Lagaris}\ \emph {et~al.}(1997)\citenamefont {Lagaris}, \citenamefont {Likas},\ and\ \citenamefont {Fotiadis}}]{LAGARIS19971}%
  \BibitemOpen
  \bibfield  {author} {\bibinfo {author} {\bibfnamefont {I.E.}\ \bibnamefont {Lagaris}}, \bibinfo {author} {\bibfnamefont {A.}~\bibnamefont {Likas}}, \ and\ \bibinfo {author} {\bibfnamefont {D.I.}\ \bibnamefont {Fotiadis}},\ }\bibfield  {title} {\enquote {\bibinfo {title} {Artificial neural network methods in quantum mechanics},}\ }\href {\doibase https://doi.org/10.1016/S0010-4655(97)00054-4} {\bibfield  {journal} {\bibinfo  {journal} {Computer Physics Communications}\ }\textbf {\bibinfo {volume} {104}},\ \bibinfo {pages} {1--14} (\bibinfo {year} {1997})}\BibitemShut {NoStop}%
\bibitem [{\citenamefont {Carleo}\ and\ \citenamefont {Troyer}(2017)}]{Carleo2017-NN_VMC}%
  \BibitemOpen
  \bibfield  {author} {\bibinfo {author} {\bibfnamefont {Giuseppe}\ \bibnamefont {Carleo}}\ and\ \bibinfo {author} {\bibfnamefont {Matthias}\ \bibnamefont {Troyer}},\ }\bibfield  {title} {\enquote {\bibinfo {title} {Solving the quantum many-body problem with artificial neural networks},}\ }\href@noop {} {\bibfield  {journal} {\bibinfo  {journal} {Science}\ }\textbf {\bibinfo {volume} {355}},\ \bibinfo {pages} {602--606} (\bibinfo {year} {2017})}\BibitemShut {NoStop}%
\bibitem [{\citenamefont {Lange}\ \emph {et~al.}(2024{\natexlab{a}})\citenamefont {Lange}, \citenamefont {Van~de Walle}, \citenamefont {Abedinnia},\ and\ \citenamefont {Bohrdt}}]{Lan24}%
  \BibitemOpen
  \bibfield  {author} {\bibinfo {author} {\bibfnamefont {Hannah}\ \bibnamefont {Lange}}, \bibinfo {author} {\bibfnamefont {Anka}\ \bibnamefont {Van~de Walle}}, \bibinfo {author} {\bibfnamefont {Atiye}\ \bibnamefont {Abedinnia}}, \ and\ \bibinfo {author} {\bibfnamefont {Annabelle}\ \bibnamefont {Bohrdt}},\ }\bibfield  {title} {\enquote {\bibinfo {title} {From architectures to applications: a review of neural quantum states},}\ }\href {\doibase 10.1088/2058-9565/ad7168} {\bibfield  {journal} {\bibinfo  {journal} {Quantum Science and Technology}\ }\textbf {\bibinfo {volume} {9}},\ \bibinfo {pages} {040501} (\bibinfo {year} {2024}{\natexlab{a}})}\BibitemShut {NoStop}%
\bibitem [{\citenamefont {Hibat-Allah}\ \emph {et~al.}(2020)\citenamefont {Hibat-Allah}, \citenamefont {Ganahl}, \citenamefont {Hayward}, \citenamefont {Melko},\ and\ \citenamefont {Carrasquilla}}]{Hibat-Allah2020-RNN}%
  \BibitemOpen
  \bibfield  {author} {\bibinfo {author} {\bibfnamefont {Mohamed}\ \bibnamefont {Hibat-Allah}}, \bibinfo {author} {\bibfnamefont {Martin}\ \bibnamefont {Ganahl}}, \bibinfo {author} {\bibfnamefont {Lauren~E.}\ \bibnamefont {Hayward}}, \bibinfo {author} {\bibfnamefont {Roger~G.}\ \bibnamefont {Melko}}, \ and\ \bibinfo {author} {\bibfnamefont {Juan}\ \bibnamefont {Carrasquilla}},\ }\bibfield  {title} {\enquote {\bibinfo {title} {Recurrent neural network wave functions},}\ }\href {\doibase 10.1103/PhysRevResearch.2.023358} {\bibfield  {journal} {\bibinfo  {journal} {Phys. Rev. Res.}\ }\textbf {\bibinfo {volume} {2}},\ \bibinfo {pages} {023358} (\bibinfo {year} {2020})}\BibitemShut {NoStop}%
\bibitem [{\citenamefont {Roth}(2020)}]{roth2020iterativeretrainingquantumspin}%
  \BibitemOpen
  \bibfield  {author} {\bibinfo {author} {\bibfnamefont {Christopher}\ \bibnamefont {Roth}},\ }\href {https://arxiv.org/abs/2003.06228} {\enquote {\bibinfo {title} {Iterative retraining of quantum spin models using recurrent neural networks},}\ } (\bibinfo {year} {2020}),\ \Eprint {http://arxiv.org/abs/2003.06228} {arXiv:2003.06228 [physics.comp-ph]} \BibitemShut {NoStop}%
\bibitem [{\citenamefont {Luo}\ \emph {et~al.}(2023)\citenamefont {Luo}, \citenamefont {Chen}, \citenamefont {Hu}, \citenamefont {Zhao}, \citenamefont {Hur},\ and\ \citenamefont {Clark}}]{PhysRevResearch.5.013216}%
  \BibitemOpen
  \bibfield  {author} {\bibinfo {author} {\bibfnamefont {Di}~\bibnamefont {Luo}}, \bibinfo {author} {\bibfnamefont {Zhuo}\ \bibnamefont {Chen}}, \bibinfo {author} {\bibfnamefont {Kaiwen}\ \bibnamefont {Hu}}, \bibinfo {author} {\bibfnamefont {Zhizhen}\ \bibnamefont {Zhao}}, \bibinfo {author} {\bibfnamefont {Vera~Mikyoung}\ \bibnamefont {Hur}}, \ and\ \bibinfo {author} {\bibfnamefont {Bryan~K.}\ \bibnamefont {Clark}},\ }\bibfield  {title} {\enquote {\bibinfo {title} {Gauge-invariant and anyonic-symmetric autoregressive neural network for quantum lattice models},}\ }\href {\doibase 10.1103/PhysRevResearch.5.013216} {\bibfield  {journal} {\bibinfo  {journal} {Phys. Rev. Res.}\ }\textbf {\bibinfo {volume} {5}},\ \bibinfo {pages} {013216} (\bibinfo {year} {2023})}\BibitemShut {NoStop}%
\bibitem [{\citenamefont {Lange}\ \emph {et~al.}(2024{\natexlab{b}})\citenamefont {Lange}, \citenamefont {Döschl}, \citenamefont {Carrasquilla},\ and\ \citenamefont {Bohrdt}}]{Lange_2024}%
  \BibitemOpen
  \bibfield  {author} {\bibinfo {author} {\bibfnamefont {Hannah}\ \bibnamefont {Lange}}, \bibinfo {author} {\bibfnamefont {Fabian}\ \bibnamefont {Döschl}}, \bibinfo {author} {\bibfnamefont {Juan}\ \bibnamefont {Carrasquilla}}, \ and\ \bibinfo {author} {\bibfnamefont {Annabelle}\ \bibnamefont {Bohrdt}},\ }\bibfield  {title} {\enquote {\bibinfo {title} {Neural network approach to quasiparticle dispersions in doped antiferromagnets},}\ }\href {\doibase 10.1038/s42005-024-01678-7} {\bibfield  {journal} {\bibinfo  {journal} {Communications Physics}\ }\textbf {\bibinfo {volume} {7}} (\bibinfo {year} {2024}{\natexlab{b}}),\ 10.1038/s42005-024-01678-7}\BibitemShut {NoStop}%
\bibitem [{\citenamefont {Melko}\ and\ \citenamefont {Carrasquilla}(2024)}]{Melko2024}%
  \BibitemOpen
  \bibfield  {author} {\bibinfo {author} {\bibfnamefont {Roger~G.}\ \bibnamefont {Melko}}\ and\ \bibinfo {author} {\bibfnamefont {Juan}\ \bibnamefont {Carrasquilla}},\ }\bibfield  {title} {\enquote {\bibinfo {title} {Language models for quantum simulation},}\ }\href {\doibase 10.1038/s43588-023-00578-0} {\bibfield  {journal} {\bibinfo  {journal} {Nature Computational Science}\ }\textbf {\bibinfo {volume} {4}},\ \bibinfo {pages} {11--18} (\bibinfo {year} {2024})}\BibitemShut {NoStop}%
\bibitem [{\citenamefont {Malica}\ \emph {et~al.}(2025)\citenamefont {Malica}, \citenamefont {Novoselov}, \citenamefont {Barnard}, \citenamefont {Kalinin}, \citenamefont {Spurgeon}, \citenamefont {Reuter}, \citenamefont {Alducin}, \citenamefont {Deringer}, \citenamefont {Csányi}, \citenamefont {Marzari}, \citenamefont {Huang}, \citenamefont {Cuniberti}, \citenamefont {Deng}, \citenamefont {Ordejón}, \citenamefont {Cole}, \citenamefont {Choudhary}, \citenamefont {Hippalgaonkar}, \citenamefont {Zhu}, \citenamefont {von Lilienfeld}, \citenamefont {Hibat-Allah}, \citenamefont {Carrasquilla}, \citenamefont {Cisotto}, \citenamefont {Zancanaro}, \citenamefont {Wenzel}, \citenamefont {Ferrari}, \citenamefont {Ustyuzhanin},\ and\ \citenamefont {Roche}}]{Malica_2025}%
  \BibitemOpen
  \bibfield  {author} {\bibinfo {author} {\bibfnamefont {Cristiano}\ \bibnamefont {Malica}}, \bibinfo {author} {\bibfnamefont {Kostya~S}\ \bibnamefont {Novoselov}}, \bibinfo {author} {\bibfnamefont {Amanda~S}\ \bibnamefont {Barnard}}, \bibinfo {author} {\bibfnamefont {Sergei~V}\ \bibnamefont {Kalinin}}, \bibinfo {author} {\bibfnamefont {Steven~R}\ \bibnamefont {Spurgeon}}, \bibinfo {author} {\bibfnamefont {Karsten}\ \bibnamefont {Reuter}}, \bibinfo {author} {\bibfnamefont {Maite}\ \bibnamefont {Alducin}}, \bibinfo {author} {\bibfnamefont {Volker~L}\ \bibnamefont {Deringer}}, \bibinfo {author} {\bibfnamefont {Gábor}\ \bibnamefont {Csányi}}, \bibinfo {author} {\bibfnamefont {Nicola}\ \bibnamefont {Marzari}}, \bibinfo {author} {\bibfnamefont {Shirong}\ \bibnamefont {Huang}}, \bibinfo {author} {\bibfnamefont {Gianaurelio}\ \bibnamefont {Cuniberti}}, \bibinfo {author} {\bibfnamefont {Qiushi}\ \bibnamefont {Deng}}, \bibinfo {author} {\bibfnamefont {Pablo}\ \bibnamefont {Ordejón}}, \bibinfo {author}
  {\bibfnamefont {Ivan}\ \bibnamefont {Cole}}, \bibinfo {author} {\bibfnamefont {Kamal}\ \bibnamefont {Choudhary}}, \bibinfo {author} {\bibfnamefont {Kedar}\ \bibnamefont {Hippalgaonkar}}, \bibinfo {author} {\bibfnamefont {Ruiming}\ \bibnamefont {Zhu}}, \bibinfo {author} {\bibfnamefont {O~Anatole}\ \bibnamefont {von Lilienfeld}}, \bibinfo {author} {\bibfnamefont {Mohamed}\ \bibnamefont {Hibat-Allah}}, \bibinfo {author} {\bibfnamefont {Juan}\ \bibnamefont {Carrasquilla}}, \bibinfo {author} {\bibfnamefont {Giulia}\ \bibnamefont {Cisotto}}, \bibinfo {author} {\bibfnamefont {Alberto}\ \bibnamefont {Zancanaro}}, \bibinfo {author} {\bibfnamefont {Wolfgang}\ \bibnamefont {Wenzel}}, \bibinfo {author} {\bibfnamefont {Andrea~C}\ \bibnamefont {Ferrari}}, \bibinfo {author} {\bibfnamefont {Andrey}\ \bibnamefont {Ustyuzhanin}}, \ and\ \bibinfo {author} {\bibfnamefont {Stephan}\ \bibnamefont {Roche}},\ }\bibfield  {title} {\enquote {\bibinfo {title} {Artificial intelligence for advanced functional materials: exploring
  current and future directions},}\ }\href {\doibase 10.1088/2515-7639/adc29d} {\bibfield  {journal} {\bibinfo  {journal} {Journal of Physics: Materials}\ }\textbf {\bibinfo {volume} {8}},\ \bibinfo {pages} {021001} (\bibinfo {year} {2025})}\BibitemShut {NoStop}%
\bibitem [{\citenamefont {Moss}\ \emph {et~al.}(2025{\natexlab{a}})\citenamefont {Moss}, \citenamefont {Wiersema}, \citenamefont {Hibat-Allah}, \citenamefont {Carrasquilla},\ and\ \citenamefont {Melko}}]{6ccd-wzhz}%
  \BibitemOpen
  \bibfield  {author} {\bibinfo {author} {\bibfnamefont {M.~Schuyler}\ \bibnamefont {Moss}}, \bibinfo {author} {\bibfnamefont {Roeland}\ \bibnamefont {Wiersema}}, \bibinfo {author} {\bibfnamefont {Mohamed}\ \bibnamefont {Hibat-Allah}}, \bibinfo {author} {\bibfnamefont {Juan}\ \bibnamefont {Carrasquilla}}, \ and\ \bibinfo {author} {\bibfnamefont {Roger~G.}\ \bibnamefont {Melko}},\ }\bibfield  {title} {\enquote {\bibinfo {title} {Leveraging recurrence in neural network wavefunctions for large-scale simulations of heisenberg antiferromagnets on the square lattice},}\ }\href {\doibase 10.1103/6ccd-wzhz} {\bibfield  {journal} {\bibinfo  {journal} {Phys. Rev. B}\ }\textbf {\bibinfo {volume} {112}},\ \bibinfo {pages} {134450} (\bibinfo {year} {2025}{\natexlab{a}})}\BibitemShut {NoStop}%
\bibitem [{\citenamefont {Moss}\ \emph {et~al.}(2025{\natexlab{b}})\citenamefont {Moss}, \citenamefont {Wiersema}, \citenamefont {Hibat-Allah}, \citenamefont {Carrasquilla},\ and\ \citenamefont {Melko}}]{nh89-6jmf}%
  \BibitemOpen
  \bibfield  {author} {\bibinfo {author} {\bibfnamefont {M.~Schuyler}\ \bibnamefont {Moss}}, \bibinfo {author} {\bibfnamefont {Roeland}\ \bibnamefont {Wiersema}}, \bibinfo {author} {\bibfnamefont {Mohamed}\ \bibnamefont {Hibat-Allah}}, \bibinfo {author} {\bibfnamefont {Juan}\ \bibnamefont {Carrasquilla}}, \ and\ \bibinfo {author} {\bibfnamefont {Roger~G.}\ \bibnamefont {Melko}},\ }\bibfield  {title} {\enquote {\bibinfo {title} {Leveraging recurrence in neural network wavefunctions for large-scale simulations of heisenberg antiferromagnets on the triangular lattice},}\ }\href {\doibase 10.1103/nh89-6jmf} {\bibfield  {journal} {\bibinfo  {journal} {Phys. Rev. B}\ }\textbf {\bibinfo {volume} {112}},\ \bibinfo {pages} {134449} (\bibinfo {year} {2025}{\natexlab{b}})}\BibitemShut {NoStop}%
\bibitem [{\citenamefont {Aboussalah}\ \emph {et~al.}(2023)\citenamefont {Aboussalah}, \citenamefont {Kwon}, \citenamefont {Patel}, \citenamefont {Chi},\ and\ \citenamefont {Lee}}]{Aboussalah2023RIM}%
  \BibitemOpen
  \bibfield  {author} {\bibinfo {author} {\bibfnamefont {Amine~Mohamed}\ \bibnamefont {Aboussalah}}, \bibinfo {author} {\bibfnamefont {Min-Jae}\ \bibnamefont {Kwon}}, \bibinfo {author} {\bibfnamefont {Raj~G.}\ \bibnamefont {Patel}}, \bibinfo {author} {\bibfnamefont {Cheng}\ \bibnamefont {Chi}}, \ and\ \bibinfo {author} {\bibfnamefont {Chi-Guhn}\ \bibnamefont {Lee}},\ }\bibfield  {title} {\enquote {\bibinfo {title} {Recursive time series data augmentation},}\ }in\ \href {https://openreview.net/forum?id=5lgD4vU-l24s} {\emph {\bibinfo {booktitle} {The Eleventh International Conference on Learning Representations}}}\ (\bibinfo {year} {2023})\BibitemShut {NoStop}%
\bibitem [{\citenamefont {Hibat-Allah}\ \emph {et~al.}(2024)\citenamefont {Hibat-Allah}, \citenamefont {Melko},\ and\ \citenamefont {Carrasquilla}}]{hibatallah2024supplementingrecurrentneuralnetwork}%
  \BibitemOpen
  \bibfield  {author} {\bibinfo {author} {\bibfnamefont {Mohamed}\ \bibnamefont {Hibat-Allah}}, \bibinfo {author} {\bibfnamefont {Roger~G.}\ \bibnamefont {Melko}}, \ and\ \bibinfo {author} {\bibfnamefont {Juan}\ \bibnamefont {Carrasquilla}},\ }\href {https://arxiv.org/abs/2207.14314} {\enquote {\bibinfo {title} {Supplementing recurrent neural network wave functions with symmetry and annealing to improve accuracy},}\ } (\bibinfo {year} {2024}),\ \Eprint {http://arxiv.org/abs/2207.14314} {arXiv:2207.14314 [cond-mat.dis-nn]} \BibitemShut {NoStop}%
\bibitem [{\citenamefont {Hibat-Allah}\ \emph {et~al.}(2025)\citenamefont {Hibat-Allah}, \citenamefont {Merali}, \citenamefont {Torlai}, \citenamefont {Melko},\ and\ \citenamefont {Carrasquilla}}]{Hibat_Allah_2025}%
  \BibitemOpen
  \bibfield  {author} {\bibinfo {author} {\bibfnamefont {Mohamed}\ \bibnamefont {Hibat-Allah}}, \bibinfo {author} {\bibfnamefont {Ejaaz}\ \bibnamefont {Merali}}, \bibinfo {author} {\bibfnamefont {Giacomo}\ \bibnamefont {Torlai}}, \bibinfo {author} {\bibfnamefont {Roger~G.}\ \bibnamefont {Melko}}, \ and\ \bibinfo {author} {\bibfnamefont {Juan}\ \bibnamefont {Carrasquilla}},\ }\bibfield  {title} {\enquote {\bibinfo {title} {Recurrent neural network wave functions for rydberg atom arrays on kagome lattice},}\ }\href {\doibase 10.1038/s42005-025-02226-7} {\bibfield  {journal} {\bibinfo  {journal} {Communications Physics}\ }\textbf {\bibinfo {volume} {8}} (\bibinfo {year} {2025}),\ 10.1038/s42005-025-02226-7}\BibitemShut {NoStop}%
\bibitem [{\citenamefont {Yang}\ \emph {et~al.}(2024)\citenamefont {Yang}, \citenamefont {Soleimanifar}, \citenamefont {Bergamaschi},\ and\ \citenamefont {Preskill}}]{YangPreskill2024-ClusterStatePaper}%
  \BibitemOpen
  \bibfield  {author} {\bibinfo {author} {\bibfnamefont {Tai-Hsuan}\ \bibnamefont {Yang}}, \bibinfo {author} {\bibfnamefont {Mehdi}\ \bibnamefont {Soleimanifar}}, \bibinfo {author} {\bibfnamefont {Thiago}\ \bibnamefont {Bergamaschi}}, \ and\ \bibinfo {author} {\bibfnamefont {John}\ \bibnamefont {Preskill}},\ }\bibfield  {title} {\enquote {\bibinfo {title} {When can classical neural networks represent quantum states?}}\ }\href@noop {} {\  (\bibinfo {year} {2024})},\ \Eprint {http://arxiv.org/abs/2410.23152} {arXiv:2410.23152 [quant-ph]} \BibitemShut {NoStop}%
\bibitem [{\citenamefont {Döschl}\ and\ \citenamefont {Bohrdt}(2025)}]{doschl2025importancecorrelationsneuralquantum}%
  \BibitemOpen
  \bibfield  {author} {\bibinfo {author} {\bibfnamefont {Fabian}\ \bibnamefont {Döschl}}\ and\ \bibinfo {author} {\bibfnamefont {Annabelle}\ \bibnamefont {Bohrdt}},\ }\href {https://arxiv.org/abs/2508.14152} {\enquote {\bibinfo {title} {Importance of correlations for neural quantum states},}\ } (\bibinfo {year} {2025}),\ \Eprint {http://arxiv.org/abs/2508.14152} {arXiv:2508.14152 [quant-ph]} \BibitemShut {NoStop}%
\bibitem [{\citenamefont {Aboussalah}\ and\ \citenamefont {Ed-dib}(2025{\natexlab{a}})}]{Aboussalah2025GNNTopology}%
  \BibitemOpen
  \bibfield  {author} {\bibinfo {author} {\bibfnamefont {Amine~Mohamed}\ \bibnamefont {Aboussalah}}\ and\ \bibinfo {author} {\bibfnamefont {Abdessalam}\ \bibnamefont {Ed-dib}},\ }\bibfield  {title} {\enquote {\bibinfo {title} {Are {GNNs} doomed by the topology of their input graph?}}\ }\href {\doibase 10.48550/arXiv.2502.17739} {\bibfield  {journal} {\bibinfo  {journal} {arXiv preprint arXiv:2502.17739}\ } (\bibinfo {year} {2025}{\natexlab{a}}),\ 10.48550/arXiv.2502.17739}\BibitemShut {NoStop}%
\bibitem [{\citenamefont {Chang}\ \emph {et~al.}(2017)\citenamefont {Chang}, \citenamefont {Zhang}, \citenamefont {Han}, \citenamefont {Yu}, \citenamefont {Guo}, \citenamefont {Tan}, \citenamefont {Cui}, \citenamefont {Witbrock}, \citenamefont {Hasegawa-Johnson},\ and\ \citenamefont {Huang}}]{Chang2017DilatedRNN}%
  \BibitemOpen
  \bibfield  {author} {\bibinfo {author} {\bibfnamefont {Shiyu}\ \bibnamefont {Chang}}, \bibinfo {author} {\bibfnamefont {Yang}\ \bibnamefont {Zhang}}, \bibinfo {author} {\bibfnamefont {Wei}\ \bibnamefont {Han}}, \bibinfo {author} {\bibfnamefont {Mo}~\bibnamefont {Yu}}, \bibinfo {author} {\bibfnamefont {Xiaoxiao}\ \bibnamefont {Guo}}, \bibinfo {author} {\bibfnamefont {Wei}\ \bibnamefont {Tan}}, \bibinfo {author} {\bibfnamefont {Xiaodong}\ \bibnamefont {Cui}}, \bibinfo {author} {\bibfnamefont {Michael}\ \bibnamefont {Witbrock}}, \bibinfo {author} {\bibfnamefont {Mark}\ \bibnamefont {Hasegawa-Johnson}}, \ and\ \bibinfo {author} {\bibfnamefont {Thomas~S.}\ \bibnamefont {Huang}},\ }\href {https://arxiv.org/abs/1710.02224} {\enquote {\bibinfo {title} {Dilated recurrent neural networks},}\ } (\bibinfo {year} {2017}),\ \Eprint {http://arxiv.org/abs/1710.02224} {arXiv:1710.02224 [cs.AI]} \BibitemShut {NoStop}%
\bibitem [{\citenamefont {Hibat-Allah}\ \emph {et~al.}(2021)\citenamefont {Hibat-Allah}, \citenamefont {Inack}, \citenamefont {Wiersema}, \citenamefont {Melko},\ and\ \citenamefont {Carrasquilla}}]{Hibat-Allah2021-VNA}%
  \BibitemOpen
  \bibfield  {author} {\bibinfo {author} {\bibfnamefont {Mohamed}\ \bibnamefont {Hibat-Allah}}, \bibinfo {author} {\bibfnamefont {Estelle~M}\ \bibnamefont {Inack}}, \bibinfo {author} {\bibfnamefont {Roeland}\ \bibnamefont {Wiersema}}, \bibinfo {author} {\bibfnamefont {Roger~G}\ \bibnamefont {Melko}}, \ and\ \bibinfo {author} {\bibfnamefont {Juan}\ \bibnamefont {Carrasquilla}},\ }\bibfield  {title} {\enquote {\bibinfo {title} {Variational neural annealing},}\ }\href@noop {} {\bibfield  {journal} {\bibinfo  {journal} {Nat. Mach. Intell.}\ }\textbf {\bibinfo {volume} {3}},\ \bibinfo {pages} {952--961} (\bibinfo {year} {2021})}\BibitemShut {NoStop}%
\bibitem [{\citenamefont {Ahsan~Khandoker}\ \emph {et~al.}(2023)\citenamefont {Ahsan~Khandoker}, \citenamefont {Munshad~Abedin},\ and\ \citenamefont {Hibat-Allah}}]{Khandoker_2023}%
  \BibitemOpen
  \bibfield  {author} {\bibinfo {author} {\bibfnamefont {Shoummo}\ \bibnamefont {Ahsan~Khandoker}}, \bibinfo {author} {\bibfnamefont {Jawaril}\ \bibnamefont {Munshad~Abedin}}, \ and\ \bibinfo {author} {\bibfnamefont {Mohamed}\ \bibnamefont {Hibat-Allah}},\ }\bibfield  {title} {\enquote {\bibinfo {title} {Supplementing recurrent neural networks with annealing to solve combinatorial optimization problems},}\ }\href {\doibase 10.1088/2632-2153/acb895} {\bibfield  {journal} {\bibinfo  {journal} {Machine Learning: Science and Technology}\ }\textbf {\bibinfo {volume} {4}},\ \bibinfo {pages} {015026} (\bibinfo {year} {2023})}\BibitemShut {NoStop}%
\bibitem [{\citenamefont {Shen}(2019)}]{shen2019mutualinformationscalingexpressive}%
  \BibitemOpen
  \bibfield  {author} {\bibinfo {author} {\bibfnamefont {Huitao}\ \bibnamefont {Shen}},\ }\href {https://arxiv.org/abs/1905.04271} {\enquote {\bibinfo {title} {Mutual information scaling and expressive power of sequence models},}\ } (\bibinfo {year} {2019}),\ \Eprint {http://arxiv.org/abs/1905.04271} {arXiv:1905.04271 [cs.LG]} \BibitemShut {NoStop}%
\bibitem [{\citenamefont {McNaughton}\ and\ \citenamefont {Hibat-Allah}(2025)}]{McNaughton2025-ku}%
  \BibitemOpen
  \bibfield  {author} {\bibinfo {author} {\bibfnamefont {Jake}\ \bibnamefont {McNaughton}}\ and\ \bibinfo {author} {\bibfnamefont {Mohamed}\ \bibnamefont {Hibat-Allah}},\ }\bibfield  {title} {\enquote {\bibinfo {title} {Adaptive neural quantum states: A recurrent neural network perspective},}\ }\href@noop {} {\  (\bibinfo {year} {2025})},\ \Eprint {http://arxiv.org/abs/2507.18700} {arXiv:2507.18700 [cond-mat.dis-nn]} \BibitemShut {NoStop}%
\bibitem [{\citenamefont {Aboussalah}\ and\ \citenamefont {Ed-dib}(2025{\natexlab{b}})}]{Aboussalah2025GeoHNN}%
  \BibitemOpen
  \bibfield  {author} {\bibinfo {author} {\bibfnamefont {Amine~Mohamed}\ \bibnamefont {Aboussalah}}\ and\ \bibinfo {author} {\bibfnamefont {Abdessalam}\ \bibnamefont {Ed-dib}},\ }\bibfield  {title} {\enquote {\bibinfo {title} {{GeoHNNs}: Geometric hamiltonian neural networks},}\ }\href {\doibase 10.48550/arXiv.2507.15678} {\bibfield  {journal} {\bibinfo  {journal} {arXiv preprint arXiv:2507.15678}\ } (\bibinfo {year} {2025}{\natexlab{b}}),\ 10.48550/arXiv.2507.15678}\BibitemShut {NoStop}%
\bibitem [{\citenamefont {Becca}\ and\ \citenamefont {Sorella}(2017)}]{becca_sorella_2017}%
  \BibitemOpen
  \bibfield  {author} {\bibinfo {author} {\bibfnamefont {Federico}\ \bibnamefont {Becca}}\ and\ \bibinfo {author} {\bibfnamefont {Sandro}\ \bibnamefont {Sorella}},\ }\href {\doibase 10.1017/9781316417041} {\emph {\bibinfo {title} {Quantum Monte Carlo Approaches for Correlated Systems}}}\ (\bibinfo  {publisher} {Cambridge University Press},\ \bibinfo {year} {2017})\BibitemShut {NoStop}%
\bibitem [{\citenamefont {Germain}\ \emph {et~al.}(2015)\citenamefont {Germain}, \citenamefont {Gregor}, \citenamefont {Murray},\ and\ \citenamefont {Larochelle}}]{germain2015mademaskedautoencoderdistribution}%
  \BibitemOpen
  \bibfield  {author} {\bibinfo {author} {\bibfnamefont {Mathieu}\ \bibnamefont {Germain}}, \bibinfo {author} {\bibfnamefont {Karol}\ \bibnamefont {Gregor}}, \bibinfo {author} {\bibfnamefont {Iain}\ \bibnamefont {Murray}}, \ and\ \bibinfo {author} {\bibfnamefont {Hugo}\ \bibnamefont {Larochelle}},\ }\href {https://arxiv.org/abs/1502.03509} {\enquote {\bibinfo {title} {Made: Masked autoencoder for distribution estimation},}\ } (\bibinfo {year} {2015}),\ \Eprint {http://arxiv.org/abs/1502.03509} {arXiv:1502.03509 [cs.LG]} \BibitemShut {NoStop}%
\bibitem [{\citenamefont {Uria}\ \emph {et~al.}(2016)\citenamefont {Uria}, \citenamefont {Côté}, \citenamefont {Gregor}, \citenamefont {Murray},\ and\ \citenamefont {Larochelle}}]{uria2016neuralautoregressivedistributionestimation}%
  \BibitemOpen
  \bibfield  {author} {\bibinfo {author} {\bibfnamefont {Benigno}\ \bibnamefont {Uria}}, \bibinfo {author} {\bibfnamefont {Marc-Alexandre}\ \bibnamefont {Côté}}, \bibinfo {author} {\bibfnamefont {Karol}\ \bibnamefont {Gregor}}, \bibinfo {author} {\bibfnamefont {Iain}\ \bibnamefont {Murray}}, \ and\ \bibinfo {author} {\bibfnamefont {Hugo}\ \bibnamefont {Larochelle}},\ }\href {https://arxiv.org/abs/1605.02226} {\enquote {\bibinfo {title} {Neural autoregressive distribution estimation},}\ } (\bibinfo {year} {2016}),\ \Eprint {http://arxiv.org/abs/1605.02226} {arXiv:1605.02226 [cs.LG]} \BibitemShut {NoStop}%
\bibitem [{\citenamefont {Wu}\ \emph {et~al.}(2019)\citenamefont {Wu}, \citenamefont {Wang},\ and\ \citenamefont {Zhang}}]{Wu_2019}%
  \BibitemOpen
  \bibfield  {author} {\bibinfo {author} {\bibfnamefont {Dian}\ \bibnamefont {Wu}}, \bibinfo {author} {\bibfnamefont {Lei}\ \bibnamefont {Wang}}, \ and\ \bibinfo {author} {\bibfnamefont {Pan}\ \bibnamefont {Zhang}},\ }\bibfield  {title} {\enquote {\bibinfo {title} {Solving statistical mechanics using variational autoregressive networks},}\ }\href {\doibase 10.1103/physrevlett.122.080602} {\bibfield  {journal} {\bibinfo  {journal} {Physical Review Letters}\ }\textbf {\bibinfo {volume} {122}} (\bibinfo {year} {2019}),\ 10.1103/physrevlett.122.080602}\BibitemShut {NoStop}%
\bibitem [{\citenamefont {Sharir}\ \emph {et~al.}(2020)\citenamefont {Sharir}, \citenamefont {Levine}, \citenamefont {Wies}, \citenamefont {Carleo},\ and\ \citenamefont {Shashua}}]{PhysRevLett.124.020503}%
  \BibitemOpen
  \bibfield  {author} {\bibinfo {author} {\bibfnamefont {Or}~\bibnamefont {Sharir}}, \bibinfo {author} {\bibfnamefont {Yoav}\ \bibnamefont {Levine}}, \bibinfo {author} {\bibfnamefont {Noam}\ \bibnamefont {Wies}}, \bibinfo {author} {\bibfnamefont {Giuseppe}\ \bibnamefont {Carleo}}, \ and\ \bibinfo {author} {\bibfnamefont {Amnon}\ \bibnamefont {Shashua}},\ }\bibfield  {title} {\enquote {\bibinfo {title} {Deep autoregressive models for the efficient variational simulation of many-body quantum systems},}\ }\href {\doibase 10.1103/PhysRevLett.124.020503} {\bibfield  {journal} {\bibinfo  {journal} {Phys. Rev. Lett.}\ }\textbf {\bibinfo {volume} {124}},\ \bibinfo {pages} {020503} (\bibinfo {year} {2020})}\BibitemShut {NoStop}%
\bibitem [{\citenamefont {Cho}\ \emph {et~al.}(2014)\citenamefont {Cho}, \citenamefont {van Merrienboer}, \citenamefont {Gulcehre}, \citenamefont {Bahdanau}, \citenamefont {Bougares}, \citenamefont {Schwenk},\ and\ \citenamefont {Bengio}}]{cho2014learningphraserepresentationsusing}%
  \BibitemOpen
  \bibfield  {author} {\bibinfo {author} {\bibfnamefont {Kyunghyun}\ \bibnamefont {Cho}}, \bibinfo {author} {\bibfnamefont {Bart}\ \bibnamefont {van Merrienboer}}, \bibinfo {author} {\bibfnamefont {Caglar}\ \bibnamefont {Gulcehre}}, \bibinfo {author} {\bibfnamefont {Dzmitry}\ \bibnamefont {Bahdanau}}, \bibinfo {author} {\bibfnamefont {Fethi}\ \bibnamefont {Bougares}}, \bibinfo {author} {\bibfnamefont {Holger}\ \bibnamefont {Schwenk}}, \ and\ \bibinfo {author} {\bibfnamefont {Yoshua}\ \bibnamefont {Bengio}},\ }\href {https://arxiv.org/abs/1406.1078} {\enquote {\bibinfo {title} {Learning phrase representations using rnn encoder-decoder for statistical machine translation},}\ } (\bibinfo {year} {2014}),\ \Eprint {http://arxiv.org/abs/1406.1078} {arXiv:1406.1078 [cs.CL]} \BibitemShut {NoStop}%
\bibitem [{\citenamefont {Sak}\ \emph {et~al.}(2014)\citenamefont {Sak}, \citenamefont {Senior},\ and\ \citenamefont {Beaufays}}]{sak2014longshorttermmemorybased}%
  \BibitemOpen
  \bibfield  {author} {\bibinfo {author} {\bibfnamefont {Haşim}\ \bibnamefont {Sak}}, \bibinfo {author} {\bibfnamefont {Andrew}\ \bibnamefont {Senior}}, \ and\ \bibinfo {author} {\bibfnamefont {Françoise}\ \bibnamefont {Beaufays}},\ }\href {https://arxiv.org/abs/1402.1128} {\enquote {\bibinfo {title} {Long short-term memory based recurrent neural network architectures for large vocabulary speech recognition},}\ } (\bibinfo {year} {2014}),\ \Eprint {http://arxiv.org/abs/1402.1128} {arXiv:1402.1128 [cs.NE]} \BibitemShut {NoStop}%
\bibitem [{\citenamefont {Sutskever}\ \emph {et~al.}(2014)\citenamefont {Sutskever}, \citenamefont {Vinyals},\ and\ \citenamefont {Le}}]{sutskever2014sequencesequencelearningneural}%
  \BibitemOpen
  \bibfield  {author} {\bibinfo {author} {\bibfnamefont {Ilya}\ \bibnamefont {Sutskever}}, \bibinfo {author} {\bibfnamefont {Oriol}\ \bibnamefont {Vinyals}}, \ and\ \bibinfo {author} {\bibfnamefont {Quoc~V.}\ \bibnamefont {Le}},\ }\href {https://arxiv.org/abs/1409.3215} {\enquote {\bibinfo {title} {Sequence to sequence learning with neural networks},}\ } (\bibinfo {year} {2014}),\ \Eprint {http://arxiv.org/abs/1409.3215} {arXiv:1409.3215 [cs.CL]} \BibitemShut {NoStop}%
\bibitem [{\citenamefont {Sch{\"a}fer}\ and\ \citenamefont {Zimmermann}(2006)}]{Schafer2006-RNN-approximators}%
  \BibitemOpen
  \bibfield  {author} {\bibinfo {author} {\bibfnamefont {Anton~Maximilian}\ \bibnamefont {Sch{\"a}fer}}\ and\ \bibinfo {author} {\bibfnamefont {Hans~Georg}\ \bibnamefont {Zimmermann}},\ }\bibfield  {title} {\enquote {\bibinfo {title} {Recurrent neural networks are universal approximators},}\ }in\ \href@noop {} {\emph {\bibinfo {booktitle} {Artificial Neural Networks -- {ICANN} 2006}}},\ \bibinfo {series and number} {Lecture notes in computer science}\ (\bibinfo  {publisher} {Springer Berlin Heidelberg},\ \bibinfo {address} {Berlin, Heidelberg},\ \bibinfo {year} {2006})\ pp.\ \bibinfo {pages} {632--640}\BibitemShut {NoStop}%
\bibitem [{\citenamefont {Lipton}\ \emph {et~al.}(2015)\citenamefont {Lipton}, \citenamefont {Berkowitz},\ and\ \citenamefont {Elkan}}]{lipton2015criticalreviewrecurrentneural}%
  \BibitemOpen
  \bibfield  {author} {\bibinfo {author} {\bibfnamefont {Zachary~C.}\ \bibnamefont {Lipton}}, \bibinfo {author} {\bibfnamefont {John}\ \bibnamefont {Berkowitz}}, \ and\ \bibinfo {author} {\bibfnamefont {Charles}\ \bibnamefont {Elkan}},\ }\href {https://arxiv.org/abs/1506.00019} {\enquote {\bibinfo {title} {A critical review of recurrent neural networks for sequence learning},}\ } (\bibinfo {year} {2015}),\ \Eprint {http://arxiv.org/abs/1506.00019} {arXiv:1506.00019 [cs.LG]} \BibitemShut {NoStop}%
\bibitem [{\citenamefont {Bravyi}(2015)}]{bravyi2015montecarlosimulationstoquastic}%
  \BibitemOpen
  \bibfield  {author} {\bibinfo {author} {\bibfnamefont {Sergey}\ \bibnamefont {Bravyi}},\ }\href {https://arxiv.org/abs/1402.2295} {\enquote {\bibinfo {title} {Monte carlo simulation of stoquastic hamiltonians},}\ } (\bibinfo {year} {2015}),\ \Eprint {http://arxiv.org/abs/1402.2295} {arXiv:1402.2295 [quant-ph]} \BibitemShut {NoStop}%
\bibitem [{\citenamefont {Zhou}\ \emph {et~al.}(2016)\citenamefont {Zhou}, \citenamefont {Wu}, \citenamefont {Zhang},\ and\ \citenamefont {Zhou}}]{zhou2016minimal}%
  \BibitemOpen
  \bibfield  {author} {\bibinfo {author} {\bibfnamefont {Guo-Bing}\ \bibnamefont {Zhou}}, \bibinfo {author} {\bibfnamefont {Jianxin}\ \bibnamefont {Wu}}, \bibinfo {author} {\bibfnamefont {Chen-Lin}\ \bibnamefont {Zhang}}, \ and\ \bibinfo {author} {\bibfnamefont {Zhi-Hua}\ \bibnamefont {Zhou}},\ }\bibfield  {title} {\enquote {\bibinfo {title} {Minimal gated unit for recurrent neural networks},}\ }\href {\doibase 10.1007/s11633-016-1006-2} {\bibfield  {journal} {\bibinfo  {journal} {International Journal of Automation and Computing}\ }\textbf {\bibinfo {volume} {13}},\ \bibinfo {pages} {226--234} (\bibinfo {year} {2016})}\BibitemShut {NoStop}%
\bibitem [{\citenamefont {Khandoker}\ \emph {et~al.}(2025)\citenamefont {Khandoker}, \citenamefont {Inack},\ and\ \citenamefont {Hibat-Allah}}]{Khandoker_2025}%
  \BibitemOpen
  \bibfield  {author} {\bibinfo {author} {\bibfnamefont {Shoummo~A}\ \bibnamefont {Khandoker}}, \bibinfo {author} {\bibfnamefont {Estelle~M}\ \bibnamefont {Inack}}, \ and\ \bibinfo {author} {\bibfnamefont {Mohamed}\ \bibnamefont {Hibat-Allah}},\ }\bibfield  {title} {\enquote {\bibinfo {title} {Lattice protein folding with variational annealing},}\ }\href {\doibase 10.1088/2632-2153/adf376} {\bibfield  {journal} {\bibinfo  {journal} {Machine Learning: Science and Technology}\ }\textbf {\bibinfo {volume} {6}},\ \bibinfo {pages} {035023} (\bibinfo {year} {2025})}\BibitemShut {NoStop}%
\bibitem [{\citenamefont {Aboussalah}\ and\ \citenamefont {Lee}(2020)}]{Aboussalah2020SDDRRL}%
  \BibitemOpen
  \bibfield  {author} {\bibinfo {author} {\bibfnamefont {Amine~Mohamed}\ \bibnamefont {Aboussalah}}\ and\ \bibinfo {author} {\bibfnamefont {Chi-Guhn}\ \bibnamefont {Lee}},\ }\bibfield  {title} {\enquote {\bibinfo {title} {Continuous control with {Stacked Deep Dynamic Recurrent Reinforcement Learning} for portfolio optimization},}\ }\href {\doibase 10.1016/j.eswa.2019.112891} {\bibfield  {journal} {\bibinfo  {journal} {Expert Systems with Applications}\ }\textbf {\bibinfo {volume} {140}},\ \bibinfo {pages} {112891} (\bibinfo {year} {2020})}\BibitemShut {NoStop}%
\bibitem [{\citenamefont {Koutník}\ \emph {et~al.}(2014)\citenamefont {Koutník}, \citenamefont {Greff}, \citenamefont {Gomez},\ and\ \citenamefont {Schmidhuber}}]{koutnik2014clockworkrnn}%
  \BibitemOpen
  \bibfield  {author} {\bibinfo {author} {\bibfnamefont {Jan}\ \bibnamefont {Koutník}}, \bibinfo {author} {\bibfnamefont {Klaus}\ \bibnamefont {Greff}}, \bibinfo {author} {\bibfnamefont {Faustino}\ \bibnamefont {Gomez}}, \ and\ \bibinfo {author} {\bibfnamefont {Jürgen}\ \bibnamefont {Schmidhuber}},\ }\href {https://arxiv.org/abs/1402.3511} {\enquote {\bibinfo {title} {A clockwork rnn},}\ } (\bibinfo {year} {2014}),\ \Eprint {http://arxiv.org/abs/1402.3511} {arXiv:1402.3511 [cs.NE]} \BibitemShut {NoStop}%
\bibitem [{\citenamefont {Chung}\ \emph {et~al.}(2017)\citenamefont {Chung}, \citenamefont {Ahn},\ and\ \citenamefont {Bengio}}]{chung2017hierarchicalmultiscalerecurrentneural}%
  \BibitemOpen
  \bibfield  {author} {\bibinfo {author} {\bibfnamefont {Junyoung}\ \bibnamefont {Chung}}, \bibinfo {author} {\bibfnamefont {Sungjin}\ \bibnamefont {Ahn}}, \ and\ \bibinfo {author} {\bibfnamefont {Yoshua}\ \bibnamefont {Bengio}},\ }\href {https://arxiv.org/abs/1609.01704} {\enquote {\bibinfo {title} {Hierarchical multiscale recurrent neural networks},}\ } (\bibinfo {year} {2017}),\ \Eprint {http://arxiv.org/abs/1609.01704} {arXiv:1609.01704 [cs.LG]} \BibitemShut {NoStop}%
\bibitem [{\citenamefont {Vaswani}\ \emph {et~al.}(2017)\citenamefont {Vaswani}, \citenamefont {Shazeer}, \citenamefont {Parmar}, \citenamefont {Uszkoreit}, \citenamefont {Jones}, \citenamefont {Gomez}, \citenamefont {Kaiser},\ and\ \citenamefont {Polosukhin}}]{NIPS2017_3f5ee243}%
  \BibitemOpen
  \bibfield  {author} {\bibinfo {author} {\bibfnamefont {Ashish}\ \bibnamefont {Vaswani}}, \bibinfo {author} {\bibfnamefont {Noam}\ \bibnamefont {Shazeer}}, \bibinfo {author} {\bibfnamefont {Niki}\ \bibnamefont {Parmar}}, \bibinfo {author} {\bibfnamefont {Jakob}\ \bibnamefont {Uszkoreit}}, \bibinfo {author} {\bibfnamefont {Llion}\ \bibnamefont {Jones}}, \bibinfo {author} {\bibfnamefont {Aidan~N}\ \bibnamefont {Gomez}}, \bibinfo {author} {\bibfnamefont {\L~ukasz}\ \bibnamefont {Kaiser}}, \ and\ \bibinfo {author} {\bibfnamefont {Illia}\ \bibnamefont {Polosukhin}},\ }\bibfield  {title} {\enquote {\bibinfo {title} {Attention is all you need},}\ }in\ \href {https://proceedings.neurips.cc/paper_files/paper/2017/file/3f5ee243547dee91fbd053c1c4a845aa-Paper.pdf} {\emph {\bibinfo {booktitle} {Advances in Neural Information Processing Systems}}},\ Vol.~\bibinfo {volume} {30},\ \bibinfo {editor} {edited by\ \bibinfo {editor} {\bibfnamefont {I.}~\bibnamefont {Guyon}}, \bibinfo {editor} {\bibfnamefont {U.~Von}\
  \bibnamefont {Luxburg}}, \bibinfo {editor} {\bibfnamefont {S.}~\bibnamefont {Bengio}}, \bibinfo {editor} {\bibfnamefont {H.}~\bibnamefont {Wallach}}, \bibinfo {editor} {\bibfnamefont {R.}~\bibnamefont {Fergus}}, \bibinfo {editor} {\bibfnamefont {S.}~\bibnamefont {Vishwanathan}}, \ and\ \bibinfo {editor} {\bibfnamefont {R.}~\bibnamefont {Garnett}}}\ (\bibinfo  {publisher} {Curran Associates, Inc.},\ \bibinfo {year} {2017})\BibitemShut {NoStop}%
\bibitem [{\citenamefont {Pascanu}\ \emph {et~al.}(2013)\citenamefont {Pascanu}, \citenamefont {Mikolov},\ and\ \citenamefont {Bengio}}]{pascanu2013difficultytrainingrecurrentneural}%
  \BibitemOpen
  \bibfield  {author} {\bibinfo {author} {\bibfnamefont {Razvan}\ \bibnamefont {Pascanu}}, \bibinfo {author} {\bibfnamefont {Tomas}\ \bibnamefont {Mikolov}}, \ and\ \bibinfo {author} {\bibfnamefont {Yoshua}\ \bibnamefont {Bengio}},\ }\href {https://arxiv.org/abs/1211.5063} {\enquote {\bibinfo {title} {On the difficulty of training recurrent neural networks},}\ } (\bibinfo {year} {2013}),\ \Eprint {http://arxiv.org/abs/1211.5063} {arXiv:1211.5063 [cs.LG]} \BibitemShut {NoStop}%
\bibitem [{\citenamefont {Vidal}(2008)}]{Vidal_2008}%
  \BibitemOpen
  \bibfield  {author} {\bibinfo {author} {\bibfnamefont {G.}~\bibnamefont {Vidal}},\ }\bibfield  {title} {\enquote {\bibinfo {title} {Class of quantum many-body states that can be efficiently simulated},}\ }\href {\doibase 10.1103/physrevlett.101.110501} {\bibfield  {journal} {\bibinfo  {journal} {Physical Review Letters}\ }\textbf {\bibinfo {volume} {101}} (\bibinfo {year} {2008}),\ 10.1103/physrevlett.101.110501}\BibitemShut {NoStop}%
\bibitem [{\citenamefont {Kingma}\ and\ \citenamefont {Ba}(2017)}]{kingma2017adammethodstochasticoptimization}%
  \BibitemOpen
  \bibfield  {author} {\bibinfo {author} {\bibfnamefont {Diederik~P.}\ \bibnamefont {Kingma}}\ and\ \bibinfo {author} {\bibfnamefont {Jimmy}\ \bibnamefont {Ba}},\ }\href {https://arxiv.org/abs/1412.6980} {\enquote {\bibinfo {title} {Adam: A method for stochastic optimization},}\ } (\bibinfo {year} {2017}),\ \Eprint {http://arxiv.org/abs/1412.6980} {arXiv:1412.6980 [cs.LG]} \BibitemShut {NoStop}%
\bibitem [{\citenamefont {Mbeng}\ \emph {et~al.}(2024)\citenamefont {Mbeng}, \citenamefont {Russomanno},\ and\ \citenamefont {Santoro}}]{Mbeng_2024}%
  \BibitemOpen
  \bibfield  {author} {\bibinfo {author} {\bibfnamefont {Glen~Bigan}\ \bibnamefont {Mbeng}}, \bibinfo {author} {\bibfnamefont {Angelo}\ \bibnamefont {Russomanno}}, \ and\ \bibinfo {author} {\bibfnamefont {Giuseppe~E.}\ \bibnamefont {Santoro}},\ }\bibfield  {title} {\enquote {\bibinfo {title} {The quantum ising chain for beginners},}\ }\href {\doibase 10.21468/scipostphyslectnotes.82} {\bibfield  {journal} {\bibinfo  {journal} {SciPost Physics Lecture Notes}\ } (\bibinfo {year} {2024}),\ 10.21468/scipostphyslectnotes.82}\BibitemShut {NoStop}%
\bibitem [{\citenamefont {Sachdev}(2011)}]{Sachdev_2011}%
  \BibitemOpen
  \bibfield  {author} {\bibinfo {author} {\bibfnamefont {Subir}\ \bibnamefont {Sachdev}},\ }\href@noop {} {\emph {\bibinfo {title} {Quantum Phase Transitions}}},\ \bibinfo {edition} {2nd}\ ed.\ (\bibinfo  {publisher} {Cambridge University Press},\ \bibinfo {year} {2011})\BibitemShut {NoStop}%
\bibitem [{\citenamefont {Di~Francesco}\ \emph {et~al.}(1997)\citenamefont {Di~Francesco}, \citenamefont {Mathieu},\ and\ \citenamefont {Sénéchal}}]{di_francesco_conformal_1997}%
  \BibitemOpen
  \bibfield  {author} {\bibinfo {author} {\bibfnamefont {Philippe}\ \bibnamefont {Di~Francesco}}, \bibinfo {author} {\bibfnamefont {Pierre}\ \bibnamefont {Mathieu}}, \ and\ \bibinfo {author} {\bibfnamefont {David}\ \bibnamefont {Sénéchal}},\ }\href {\doibase 10.1007/978-1-4612-2256-9} {\emph {\bibinfo {title} {Conformal {Field} {Theory}}}},\ Graduate {Texts} in {Contemporary} {Physics}\ (\bibinfo  {publisher} {Springer New York},\ \bibinfo {address} {New York, NY},\ \bibinfo {year} {1997})\BibitemShut {NoStop}%
\bibitem [{\citenamefont {Alba}\ \emph {et~al.}(2010)\citenamefont {Alba}, \citenamefont {Tagliacozzo},\ and\ \citenamefont {Calabrese}}]{Alba_2010}%
  \BibitemOpen
  \bibfield  {author} {\bibinfo {author} {\bibfnamefont {Vincenzo}\ \bibnamefont {Alba}}, \bibinfo {author} {\bibfnamefont {Luca}\ \bibnamefont {Tagliacozzo}}, \ and\ \bibinfo {author} {\bibfnamefont {Pasquale}\ \bibnamefont {Calabrese}},\ }\bibfield  {title} {\enquote {\bibinfo {title} {Entanglement entropy of two disjoint blocks in critical ising models},}\ }\href {\doibase 10.1103/physrevb.81.060411} {\bibfield  {journal} {\bibinfo  {journal} {Physical Review B}\ }\textbf {\bibinfo {volume} {81}} (\bibinfo {year} {2010}),\ 10.1103/physrevb.81.060411}\BibitemShut {NoStop}%
\bibitem [{\citenamefont {Calabrese}\ and\ \citenamefont {Cardy}(2004)}]{Pasquale_Calabrese_2004}%
  \BibitemOpen
  \bibfield  {author} {\bibinfo {author} {\bibfnamefont {Pasquale}\ \bibnamefont {Calabrese}}\ and\ \bibinfo {author} {\bibfnamefont {John}\ \bibnamefont {Cardy}},\ }\bibfield  {title} {\enquote {\bibinfo {title} {Entanglement entropy and quantum field theory},}\ }\href {\doibase 10.1088/1742-5468/2004/06/p06002} {\bibfield  {journal} {\bibinfo  {journal} {Journal of Statistical Mechanics: Theory and Experiment}\ }\textbf {\bibinfo {volume} {2004}},\ \bibinfo {pages} {P06002} (\bibinfo {year} {2004})}\BibitemShut {NoStop}%
\bibitem [{\citenamefont {Ju}\ \emph {et~al.}(2012)\citenamefont {Ju}, \citenamefont {Kallin}, \citenamefont {Fendley}, \citenamefont {Hastings},\ and\ \citenamefont {Melko}}]{PhysRevB.85.165121}%
  \BibitemOpen
  \bibfield  {author} {\bibinfo {author} {\bibfnamefont {Hyejin}\ \bibnamefont {Ju}}, \bibinfo {author} {\bibfnamefont {Ann~B.}\ \bibnamefont {Kallin}}, \bibinfo {author} {\bibfnamefont {Paul}\ \bibnamefont {Fendley}}, \bibinfo {author} {\bibfnamefont {Matthew~B.}\ \bibnamefont {Hastings}}, \ and\ \bibinfo {author} {\bibfnamefont {Roger~G.}\ \bibnamefont {Melko}},\ }\bibfield  {title} {\enquote {\bibinfo {title} {Entanglement scaling in two-dimensional gapless systems},}\ }\href {\doibase 10.1103/PhysRevB.85.165121} {\bibfield  {journal} {\bibinfo  {journal} {Phys. Rev. B}\ }\textbf {\bibinfo {volume} {85}},\ \bibinfo {pages} {165121} (\bibinfo {year} {2012})}\BibitemShut {NoStop}%
\bibitem [{\citenamefont {Holzhey}\ \emph {et~al.}(1994)\citenamefont {Holzhey}, \citenamefont {Larsen},\ and\ \citenamefont {Wilczek}}]{HOLZHEY1994443}%
  \BibitemOpen
  \bibfield  {author} {\bibinfo {author} {\bibfnamefont {Christoph}\ \bibnamefont {Holzhey}}, \bibinfo {author} {\bibfnamefont {Finn}\ \bibnamefont {Larsen}}, \ and\ \bibinfo {author} {\bibfnamefont {Frank}\ \bibnamefont {Wilczek}},\ }\bibfield  {title} {\enquote {\bibinfo {title} {Geometric and renormalized entropy in conformal field theory},}\ }\href {\doibase https://doi.org/10.1016/0550-3213(94)90402-2} {\bibfield  {journal} {\bibinfo  {journal} {Nuclear Physics B}\ }\textbf {\bibinfo {volume} {424}},\ \bibinfo {pages} {443--467} (\bibinfo {year} {1994})}\BibitemShut {NoStop}%
\bibitem [{\citenamefont {Korepin}(2004)}]{PhysRevLett.92.096402}%
  \BibitemOpen
  \bibfield  {author} {\bibinfo {author} {\bibfnamefont {V.~E.}\ \bibnamefont {Korepin}},\ }\bibfield  {title} {\enquote {\bibinfo {title} {Universality of entropy scaling in one dimensional gapless models},}\ }\href {\doibase 10.1103/PhysRevLett.92.096402} {\bibfield  {journal} {\bibinfo  {journal} {Phys. Rev. Lett.}\ }\textbf {\bibinfo {volume} {92}},\ \bibinfo {pages} {096402} (\bibinfo {year} {2004})}\BibitemShut {NoStop}%
\bibitem [{\citenamefont {Jreissaty}\ \emph {et~al.}(2026)\citenamefont {Jreissaty}, \citenamefont {Zhang}, \citenamefont {Quijano}, \citenamefont {Carrasquilla},\ and\ \citenamefont {Wiersema}}]{Jreissaty_2026}%
  \BibitemOpen
  \bibfield  {author} {\bibinfo {author} {\bibfnamefont {Andrew}\ \bibnamefont {Jreissaty}}, \bibinfo {author} {\bibfnamefont {Hang}\ \bibnamefont {Zhang}}, \bibinfo {author} {\bibfnamefont {Jairo~C.}\ \bibnamefont {Quijano}}, \bibinfo {author} {\bibfnamefont {Juan}\ \bibnamefont {Carrasquilla}}, \ and\ \bibinfo {author} {\bibfnamefont {Roeland}\ \bibnamefont {Wiersema}},\ }\bibfield  {title} {\enquote {\bibinfo {title} {Entanglement and optimization within autoregressive neural quantum states},}\ }\href {\doibase 10.1103/t2cg-kr7y} {\bibfield  {journal} {\bibinfo  {journal} {Physical Review Research}\ }\textbf {\bibinfo {volume} {8}} (\bibinfo {year} {2026}),\ 10.1103/t2cg-kr7y}\BibitemShut {NoStop}%
\bibitem [{\citenamefont {Rende}\ \emph {et~al.}(2024)\citenamefont {Rende}, \citenamefont {Viteritti}, \citenamefont {Bardone}, \citenamefont {Becca},\ and\ \citenamefont {Goldt}}]{Rende_2024}%
  \BibitemOpen
  \bibfield  {author} {\bibinfo {author} {\bibfnamefont {Riccardo}\ \bibnamefont {Rende}}, \bibinfo {author} {\bibfnamefont {Luciano~Loris}\ \bibnamefont {Viteritti}}, \bibinfo {author} {\bibfnamefont {Lorenzo}\ \bibnamefont {Bardone}}, \bibinfo {author} {\bibfnamefont {Federico}\ \bibnamefont {Becca}}, \ and\ \bibinfo {author} {\bibfnamefont {Sebastian}\ \bibnamefont {Goldt}},\ }\bibfield  {title} {\enquote {\bibinfo {title} {A simple linear algebra identity to optimize large-scale neural network quantum states},}\ }\href {\doibase 10.1038/s42005-024-01732-4} {\bibfield  {journal} {\bibinfo  {journal} {Communications Physics}\ }\textbf {\bibinfo {volume} {7}} (\bibinfo {year} {2024}),\ 10.1038/s42005-024-01732-4}\BibitemShut {NoStop}%
\bibitem [{\citenamefont {Chen}\ and\ \citenamefont {Heyl}(2024)}]{Chen2024}%
  \BibitemOpen
  \bibfield  {author} {\bibinfo {author} {\bibfnamefont {Ao}~\bibnamefont {Chen}}\ and\ \bibinfo {author} {\bibfnamefont {Markus}\ \bibnamefont {Heyl}},\ }\bibfield  {title} {\enquote {\bibinfo {title} {Empowering deep neural quantum states through efficient optimization},}\ }\href {\doibase 10.1038/s41567-024-02566-1} {\bibfield  {journal} {\bibinfo  {journal} {Nature Physics}\ }\textbf {\bibinfo {volume} {20}},\ \bibinfo {pages} {1476--1481} (\bibinfo {year} {2024})}\BibitemShut {NoStop}%
\bibitem [{\citenamefont {Armegioiu}\ \emph {et~al.}(2025)\citenamefont {Armegioiu}, \citenamefont {Carrasquilla}, \citenamefont {Mishra}, \citenamefont {Müller}, \citenamefont {Nys}, \citenamefont {Zeinhofer},\ and\ \citenamefont {Zhang}}]{armegioiu2025functionalneuralwavefunctionoptimization}%
  \BibitemOpen
  \bibfield  {author} {\bibinfo {author} {\bibfnamefont {Victor}\ \bibnamefont {Armegioiu}}, \bibinfo {author} {\bibfnamefont {Juan}\ \bibnamefont {Carrasquilla}}, \bibinfo {author} {\bibfnamefont {Siddhartha}\ \bibnamefont {Mishra}}, \bibinfo {author} {\bibfnamefont {Johannes}\ \bibnamefont {Müller}}, \bibinfo {author} {\bibfnamefont {Jannes}\ \bibnamefont {Nys}}, \bibinfo {author} {\bibfnamefont {Marius}\ \bibnamefont {Zeinhofer}}, \ and\ \bibinfo {author} {\bibfnamefont {Hang}\ \bibnamefont {Zhang}},\ }\href {https://arxiv.org/abs/2507.10835} {\enquote {\bibinfo {title} {Functional neural wavefunction optimization},}\ } (\bibinfo {year} {2025}),\ \Eprint {http://arxiv.org/abs/2507.10835} {arXiv:2507.10835 [cond-mat.str-el]} \BibitemShut {NoStop}%
\bibitem [{\citenamefont {Eisert}\ \emph {et~al.}(2010)\citenamefont {Eisert}, \citenamefont {Cramer},\ and\ \citenamefont {Plenio}}]{Eisert_2010}%
  \BibitemOpen
  \bibfield  {author} {\bibinfo {author} {\bibfnamefont {J.}~\bibnamefont {Eisert}}, \bibinfo {author} {\bibfnamefont {M.}~\bibnamefont {Cramer}}, \ and\ \bibinfo {author} {\bibfnamefont {M.~B.}\ \bibnamefont {Plenio}},\ }\bibfield  {title} {\enquote {\bibinfo {title} {Colloquium: Area laws for the entanglement entropy},}\ }\href {\doibase 10.1103/revmodphys.82.277} {\bibfield  {journal} {\bibinfo  {journal} {Reviews of Modern Physics}\ }\textbf {\bibinfo {volume} {82}},\ \bibinfo {pages} {277–306} (\bibinfo {year} {2010})}\BibitemShut {NoStop}%
\bibitem [{\citenamefont {Bernien}\ \emph {et~al.}(2017)\citenamefont {Bernien}, \citenamefont {Schwartz}, \citenamefont {Keesling}, \citenamefont {Levine}, \citenamefont {Omran}, \citenamefont {Pichler}, \citenamefont {Choi}, \citenamefont {Zibrov}, \citenamefont {Endres}, \citenamefont {Greiner}, \citenamefont {Vuleti{\'{c}}},\ and\ \citenamefont {Lukin}}]{Bernien2017}%
  \BibitemOpen
  \bibfield  {author} {\bibinfo {author} {\bibfnamefont {Hannes}\ \bibnamefont {Bernien}}, \bibinfo {author} {\bibfnamefont {Sylvain}\ \bibnamefont {Schwartz}}, \bibinfo {author} {\bibfnamefont {Alexander}\ \bibnamefont {Keesling}}, \bibinfo {author} {\bibfnamefont {Harry}\ \bibnamefont {Levine}}, \bibinfo {author} {\bibfnamefont {Ahmed}\ \bibnamefont {Omran}}, \bibinfo {author} {\bibfnamefont {Hannes}\ \bibnamefont {Pichler}}, \bibinfo {author} {\bibfnamefont {Soonwon}\ \bibnamefont {Choi}}, \bibinfo {author} {\bibfnamefont {Alexander~S.}\ \bibnamefont {Zibrov}}, \bibinfo {author} {\bibfnamefont {Manuel}\ \bibnamefont {Endres}}, \bibinfo {author} {\bibfnamefont {Markus}\ \bibnamefont {Greiner}}, \bibinfo {author} {\bibfnamefont {Vladan}\ \bibnamefont {Vuleti{\'{c}}}}, \ and\ \bibinfo {author} {\bibfnamefont {Mikhail~D.}\ \bibnamefont {Lukin}},\ }\bibfield  {title} {\enquote {\bibinfo {title} {Probing many-body dynamics on a 51-atom quantum simulator},}\ }\href {\doibase 10.1038/nature24622} {\bibfield
  {journal} {\bibinfo  {journal} {Nature}\ }\textbf {\bibinfo {volume} {551}},\ \bibinfo {pages} {579--584} (\bibinfo {year} {2017})}\BibitemShut {NoStop}%
\bibitem [{\citenamefont {Zhang}\ \emph {et~al.}(2017)\citenamefont {Zhang}, \citenamefont {Pagano}, \citenamefont {Hess}, \citenamefont {Kyprianidis}, \citenamefont {Becker}, \citenamefont {Kaplan}, \citenamefont {Gorshkov}, \citenamefont {Gong},\ and\ \citenamefont {Monroe}}]{Zhang2017}%
  \BibitemOpen
  \bibfield  {author} {\bibinfo {author} {\bibfnamefont {J.}~\bibnamefont {Zhang}}, \bibinfo {author} {\bibfnamefont {G.}~\bibnamefont {Pagano}}, \bibinfo {author} {\bibfnamefont {P.~W.}\ \bibnamefont {Hess}}, \bibinfo {author} {\bibfnamefont {A.}~\bibnamefont {Kyprianidis}}, \bibinfo {author} {\bibfnamefont {P.}~\bibnamefont {Becker}}, \bibinfo {author} {\bibfnamefont {H.}~\bibnamefont {Kaplan}}, \bibinfo {author} {\bibfnamefont {A.~V.}\ \bibnamefont {Gorshkov}}, \bibinfo {author} {\bibfnamefont {Z.-X.}\ \bibnamefont {Gong}}, \ and\ \bibinfo {author} {\bibfnamefont {C.}~\bibnamefont {Monroe}},\ }\bibfield  {title} {\enquote {\bibinfo {title} {Observation of a many-body dynamical phase transition with a 53-qubit quantum simulator},}\ }\href {\doibase 10.1038/nature24654} {\bibfield  {journal} {\bibinfo  {journal} {Nature}\ }\textbf {\bibinfo {volume} {551}},\ \bibinfo {pages} {601--604} (\bibinfo {year} {2017})}\BibitemShut {NoStop}%
\bibitem [{\citenamefont {Flajolet}\ and\ \citenamefont {Sedgewick}(2009)}]{flajolet2009analytic}%
  \BibitemOpen
  \bibfield  {author} {\bibinfo {author} {\bibfnamefont {Philippe}\ \bibnamefont {Flajolet}}\ and\ \bibinfo {author} {\bibfnamefont {Robert}\ \bibnamefont {Sedgewick}},\ }\href@noop {} {\emph {\bibinfo {title} {Analytic combinatorics}}}\ (\bibinfo  {publisher} {cambridge University press},\ \bibinfo {year} {2009})\BibitemShut {NoStop}%
\end{thebibliography}%

\end{document}